\pdfoutput=1
\documentclass[runningheads, orivec]{llncs}
\usepackage{etex}
\usepackage{graphicx}
\usepackage[hyphens]{url}
\usepackage[hidelinks=true]{hyperref}  
\usepackage[table,xcdraw,dvipsnames]{xcolor}
\usepackage[linesnumbered,boxed]{algorithm2e}
\usepackage{multirow,tabularx}
\usepackage{booktabs, makecell}
\usepackage[font=small,labelfont=bf]{caption}
\usepackage{subcaption}
\usepackage{longtable}
\usepackage{rotating}
\usepackage{amsmath, amssymb, amsthm}
\usepackage{mathabx}
\usepackage{mathtools}
\usepackage{pifont}
\usepackage[integrals]{wasysym}
\usepackage{wrapfig}
\usepackage{tikz}
\usetikzlibrary{automata,positioning,arrows}
\usetikzlibrary{decorations.markings}
\usetikzlibrary{calc,fit}
\usepackage{verbatim}
\usepackage{proof}
\usepackage{float}
\usepackage{dsfont}
\usepackage{textcomp}
\usepackage{rotating}
\usepackage{paralist}
\usepackage[shortlabels]{enumitem}
\usepackage[strings]{underscore}
\usepackage{colortbl}
\newboolean{SHORT}
\setboolean{SHORT}{false}
\newcommand{\longversion}[1]{\ifthenelse{\boolean{SHORT}}{}{#1}}%
\newcommand{\shortversion}[1]{\ifthenelse{\boolean{SHORT}}{#1}{}}%
 
\newcommand*\blacklaptop{\includegraphics[height=\heightof{M}]{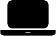}}
\newcommand*\whitelaptop{\includegraphics[height=\heightof{M}]{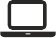}}
\newlist{longenum}{enumerate}{5}
\setlist[longenum,1]{label=\roman*)}
\setlist[longenum,2]{label=\alph*)}
\setlist[longenum,3]{label=\arabic*)}
\setlist[longenum,4]{label=(\roman*)}
\setlist[longenum,5]{label=(\alph*)}

\newtheoremstyle{mystyle}
{}
{}
{\itshape}
{}
{\bfseries}
{.}
{ }
{\thmname{#1}\thmnumber{ #2}\thmnote{ (#3)}}

\theoremstyle{mystyle}

\newtheorem{thm}{Theorem}
 
\newtheorem{defn}{Definition} 

\usepackage{stmaryrd}
\usepackage{trimclip}

\makeatletter
\DeclareRobustCommand{\shortto}{%
	\mathrel{\mathpalette\short@to\relax}%
}

\newcommand{\short@to}[2]{%
	\mkern2mu
	\clipbox{{.5\width} 0 0 0}{$\m@th#1\vphantom{+}{\shortrightarrow}$}%
}
\makeatother

\newcommand{\toolname}{MOCS}

\makeatletter
\tikzset{%
	remember picture with id/.style={%
		remember picture,
		overlay,
		save picture id=#1,
	},
	save picture id/.code={%
		\edef\pgf@temp{#1}%
		\immediate\write\pgfutil@auxout{%
			\noexpand\savepointas{\pgf@temp}{\pgfpictureid}}%
	},
	if picture id/.code args={#1#2#3}{%
		\@ifundefined{save@pt@#1}{%
			\pgfkeysalso{#3}%
		}{
			\pgfkeysalso{#2}%
		}
	}
}

\def\savepointas#1#2{%
	\expandafter\gdef\csname save@pt@#1\endcsname{#2}%
}

\def\tmk@labeldef#1,#2\@nil{%
	\def\tmk@label{#1}%
	\def\tmk@def{#2}%
}

\tikzdeclarecoordinatesystem{pic}{%
	\pgfutil@in@,{#1}%
	\ifpgfutil@in@%
	\tmk@labeldef#1\@nil
	\else
	\tmk@labeldef#1,(0pt,0pt)\@nil
	\fi
	\@ifundefined{save@pt@\tmk@label}{%
		\tikz@scan@one@point\pgfutil@firstofone\tmk@def
	}{%
		\pgfsys@getposition{\csname save@pt@\tmk@label\endcsname}\save@orig@pic%
		\pgfsys@getposition{\pgfpictureid}\save@this@pic%
		\pgf@process{\pgfpointorigin\save@this@pic}%
		\pgf@xa=\pgf@x
		\pgf@ya=\pgf@y
		\pgf@process{\pgfpointorigin\save@orig@pic}%
		\advance\pgf@x by -\pgf@xa
		\advance\pgf@y by -\pgf@ya
	}%
}
\newcommand\tikzmark[2][]{%
	\tikz[remember picture with id=#2] #1;}
\makeatother

\newcommand\MyBox[4][-1ex]{%
	\tikz[remember picture,overlay,pin distance=0cm]
	{\draw[draw=white,line width=1pt,fill=#4!15,rectangle,rounded corners]
		( $ (pic cs:#2) + (-1ex,2ex) $ ) rectangle ( $ (pic cs:#3) + (1ex,#1) $ );
	}
}

\newcommand\stutt{\mathrel{\overset{\makebox[0pt]{\mbox{\normalfont\tiny\sffamily st}}}{\equiv}}}

\newcommand{\hookrightarrowdbl}{\hookrightarrow\mathrel{\mkern-14mu}\rightarrow}

\newcommand{\xhookrightarrowdbl}[2][]{%
	\xhookrightarrow[#1]{#2}\mathrel{\mkern-14mu}\rightarrow
}

\newtheorem{innercustomgeneric}{\customgenericname}
\providecommand{\customgenericname}{}
\newcommand{\newcustomtheorem}[2]{%
	\newenvironment{#1}[1]
	{%
		\renewcommand\customgenericname{#2}%
		\renewcommand\theinnercustomgeneric{##1}%
		\innercustomgeneric
	}
	{\endinnercustomgeneric}
}

\newcustomtheorem{customthm}{Theorem}
\newcustomtheorem{customlemma}{Lemma}
\newcommand{\ignore}[1]{}

\newcommand\subtrace{\ooalign{$-$\cr$<$}}

\SetAlCapSkip{.7ex}

\begin{document}
	\title{Towards Model Checking Real-World Software-Defined Networks\\\small{(version with appendix)}}
	
	\author{Vasileios Klimis \and George Parisis \and  Bernhard Reus }
	\authorrunning{V. Klimis, G. Parisis, and B. Reus}
	\titlerunning{Towards Model Checking Real-World Software-Defined Networks}
	%
	\institute{University of Sussex, UK\\
		\email{\{v.klimis, g.parisis, bernhard\}@sussex.ac.uk}
	}

	\maketitle              
\begin{abstract}

In software-defined networks (SDN), a controller program is in charge of deploying diverse network functionality across a large number of switches, but this comes at a great risk: deploying buggy controller code could result in network and service disruption and security loopholes. The automatic detection of bugs or, even better, verification of their absence is thus most desirable, yet the size of the network and the complexity of the controller makes this a challenging undertaking.
In this paper, we propose \toolname, a highly expressive, optimised SDN model that allows capturing subtle real-world bugs, in a reasonable amount of time. This is achieved by (1) analysing the model for possible partial order reductions, (2) statically pre-computing packet equivalence classes and (3) indexing packets and rules that exist in the model. We demonstrate its superiority compared to the state of the art in terms of expressivity, by providing examples of realistic bugs that a prototype implementation of \toolname\ in {\sc Uppaal} caught, and performance/scalability, by running examples on various sizes of network topologies, highlighting the importance of our abstractions and optimisations.\\\\
\emph{Note: This is an extended version of our paper (with the same name), which appears in CAV 2020.}

\end{abstract}

\section{Introduction}
\label{introduction}

Software-Defined Networking (SDN) \cite{TheRoadToSDN} has brought about a paradigm shift in designing and operating computer networks. A logically centralised controller implements the control logic and `programs' the data plane, which is defined by flow tables installed in network switches. SDN enables the rapid development of advanced and diverse network functionality; e.g. in designing next-generation inter-data centre traffic engineering \cite{DevoFlow}, load balancing \cite{Plug-n-Serve}, firewalls \cite{SDNFirewall}, and Internet exchange points (IXPs) \cite{SDX}. SDN has gained noticeable ground in the industry, with major vendors integrating OpenFlow \cite{Openflow}, the de-facto SDN standard maintained by the Open Networking Forum, in their products. Operators deploy it at scale \cite{Google,Ananta}. SDN presents a unique opportunity for innovation and rapid development of complex network services by enabling all players, not just vendors, to develop and deploy control and data plane functionality in networks. This comes at a great risk; deploying buggy code at the controller could result in problematic flow entries at the data plane and, potentially, service disruption \cite{ERA,Batfish,NetTroub,vision} and security loopholes \cite{OF-RHM,DDoS}. Understanding and fixing such bugs is far from trivial, given the distributed and concurrent nature of computer networks and the complexity of the control plane \cite{shenker}.

With the advent of SDN, a large body of research on verifying network properties has emerged \cite{survey_NV}. Static network analysis approaches \cite{Anteater,HSA,NetSAT,FlowChecker,Flover,Dobrescu} can only verify network properties on a given fixed network configuration but this may be changing very quickly (e.g. as in \cite{hedera}). Another key limitation is the fact that they cannot reason about the controller program, which, itself, is responsible for the changes in the network configuration. Dynamic approaches, such as \cite{VeriFlow,Plotkin,Libra,Deltanet,Netplumber,AtomicPredicates}, are able to reason about network properties as changes happen (i.e. as flow entries in switches' flow tables are being added and deleted), but they cannot reason about the controller program either. As a result, when a property violation is detected, there is no straightforward way to fix the bug in the controller code, as these systems are oblivious of the running code. Identifying bugs in large and complex deployments can be extremely challenging.

Formal verification methods that include the controller code in the model of the network can solve this important problem. Symbolic execution methods, such as \cite{SymNet,Nice,Dobrescu,NetSMC,BUZZ,VeriCon2,SyNET}, evaluate programs using symbolic variables accumulating path-conditions along the way that then can be solved logically. However, they suffer from the path explosion problem caused by loops and function calls which means verification does not scale to larger controller programs (bug finding still works but is limited). Model checking SDNs is a promising area even though only few studies have been undertaken \cite{NetSMC,SDNactors,Nice,Sethi,Kuai,McClurg}. Networks and controller can be naturally modelled as transition systems. State explosion is always a problem but can be mitigated by using abstraction and optimisation techniques (i.e. partial order reductions). At the same time, modern model checkers \cite{Spin,UPPAAL,NuSMV,Alloy,JAVAPathFinder} are very efficient.

{\sc n}et{\sc smc} \cite{NetSMC} uses a bespoke \emph{symbolic} model checking algorithm for checking properties given a subset of computation tree logic that allows quantification only over all paths. As a result, this approach scales relatively well, but the requirement that only one packet can travel through the network at any time is very restrictive and ignores race conditions. {\sc nice} \cite{Nice} employs model checking but only looks at a limited amount of input packets that are extracted through symbolically executing the controller code. As a result, it is a bug-finding tool only. The authors in \cite{Sethi} propose a model checking approach that can deal with dynamic controller updates and an arbitrary number of packets but require manually inserted non-interference lemmas that constrain the set of packets that can appear in the network. This significantly limits its applicability in realistic network deployments. Kuai \cite{Kuai} overcomes this limitation by introducing model-specific partial order reductions (PORs) that result in pruning the state space by avoiding redundant explorations. However, it has limitations explained at the end of this section. 

In this paper, we take a step further towards the full realisation of model checking real-world SDNs by introducing \toolname\ (MOdel Checking for Software defined networks)\footnote{A release of \toolname\ is publicly available at \url{https://tinyurl.com/y95qtv5k}}, a highly expressive, optimised SDN model which we implemented in {\sc Uppaal}\footnote{{\sc Uppaal} has been chosen as future plans include extending the model to timed actions like e.g.\ timeouts. Note that the model can be implemented in any model checker.}
 \cite{UPPAAL}. \toolname, compared to the state of the art in model checking SDNs, can model network behaviour more realistically and verify larger deployments using fewer resources. The main contributions of this paper are:

\noindent\textbf{Model Generality}. The proposed network model is closer to the OpenFlow standard than previous models (e.g.\ \cite{Kuai})
to reflect commonly exhibited behaviour between the controller and network switches.
 More specifically, it allows for race conditions between control messages and includes a significant number of  OpenFlow interactions, including barrier response messages. In our experimentation section, we present families of elusive bugs that can be efficiently captured by \toolname. 
	
\noindent\textbf{Model Checking Optimisations}. To tackle the state explosion problem we propose context-dependent \emph{partial order reductions}  by considering the concrete control program and specification in question. We establish the soundness of the proposed optimisations\shortversion{\footnote{Proofs appear in an extended version \url{http://sro.sussex.ac.uk/id/eprint/87495}.}}. Moreover, we propose \emph{state representation optimisations}, namely packet and rule indexing, identification of packet equivalence classes and bit packing, to improve performance. We evaluate the benefits from all proposed optimisations in \S \ref{experimental-evaluation}.\\
 
Our model has been inspired by Kuai \cite{Kuai}. According to the contributions above, however, we consider \toolname\ to be a considerable improvement. We model more OpenFlow messages and interactions, enabling us to check for bugs that \cite{Kuai} cannot even express (see discussion in \S\ref{subsec:expressivity}). Our context-dependent PORs systematically explore possibilities for optimisation. Our optimisation techniques still allow \toolname\ to run at least as efficiently as Kuai, often with even better performance.

\section{Software-Defined Network Model}
\label{sec:modelSDN}

A key objective of our work is to enable the verification of network-wide properties in real-world SDNs. In order to fulfil this ambition, we present an extended network model to capture complex interactions between the SDN controller and the network. Below we describe the adopted network model, its state and transitions.

\subsection{Formal Model Definition}
\label{semantics}

The formal definition of the proposed SDN model is by means of an action-deterministic transition system. We parameterise the model by the underlying network topology $\lambda$ and the controller program {\sc cp} in use, as explained further below (\S\ref{model-components}).

\begin{defn}
\label{def:SDNmodel}
 An SDN model is a 6-tuple
 $\mathcal{M}_{(\lambda,\textsc{cp})} = (S, s_0, A, \hookrightarrow, AP, L)$, where $S$ is the set of all states the SDN may enter, $s_0$ the initial state, $A$ the set of actions which encode the events the network may engage in, 
$\hookrightarrow \subseteq S \times A \times S$ the transition relation describing which execution steps the system undergoes as it perform actions, 
$AP$ a set of atomic propositions describing relevant state properties, and $L: S \to 2^{AP}$ is a labelling function, which relates to any state $s\in S$ a set $L(s) \in 2^{AP}$ of those atomic propositions that are true for $s$. Such an SDN model is composed of several smaller systems, which model network components (hosts, switches and the controller) that communicate via queues and, combined, give rise to the definition of $\hookrightarrow$. The states of an SDN transition system are 3-tuples $(\pi, \delta, \gamma)$, where $\pi$ represents the state of each host, $\delta$ the state of each switch, and $\gamma$ the controller state. The components are explained in \S\ref{model-components} and the transitions $\hookrightarrow$ in \S\ref{transitions}.
\end{defn}
Figure~\ref{model} illustrates a high-level view of OpenFlow interactions (left side), modelled actions and queues (right side). 

\begin{figure}[h]
	\captionsetup{width=.9\textwidth}
	\begin{center}
		\includegraphics[scale=1]{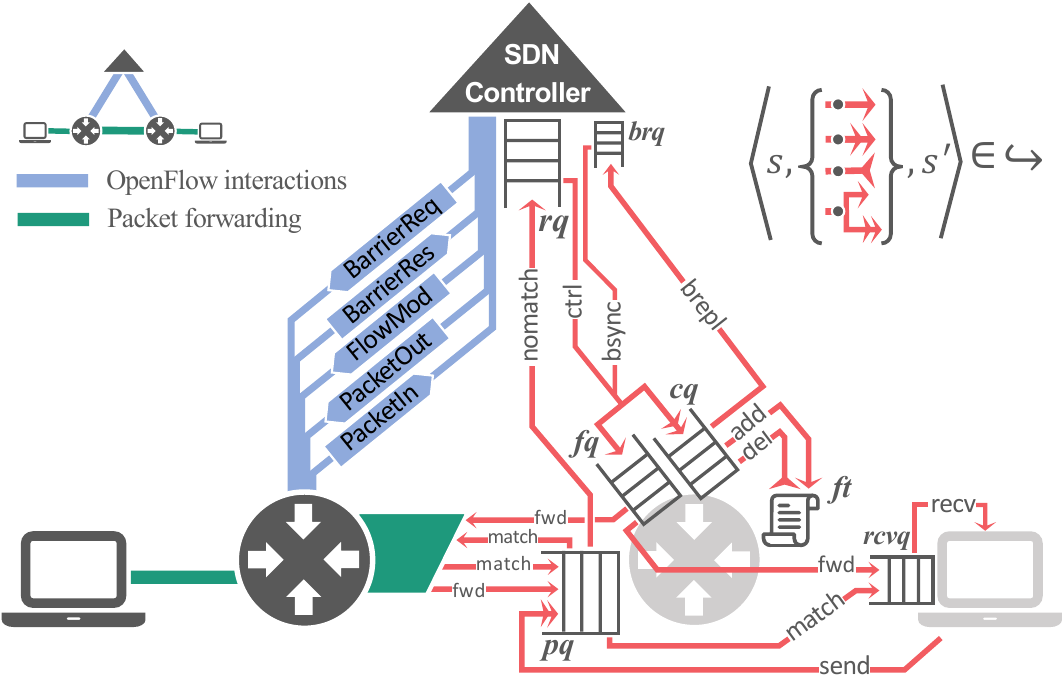}
		\caption{A high-level view of OpenFlow interactions using OpenFlow specification terminology (left half) and the modelled actions (right half). A red solid-line arrow depicts an action which, when fired, (1) dequeues an item from the queue the arrow begins at, and (2) adds an item in the queue the arrowhead points to (or multiple items if the arrow is double-headed). Deleting an item from the target queue is denoted by a reverse arrowhead. A forked arrow denotes multiple targeted queues.}
		\label{model}
	\end{center}
\end{figure}

\subsection{SDN Model Components}
\label{model-components}
Throughout we will use the common `dot-notation' (\texttt{\_.\_}) to refer to components of composite gadgets (tuples), e.g.\ queues of switches, or parts of the state. We use obvious names for the projections functions like $s.\delta.sw.pq$ for the packet queue of the switch \emph{sw} in state $s$. At times we will also use $t_1$ and $t_2$ for the first and second projection of tuple $t$.

\noindent\textbf{Network Topology. }A location $\mathit{(n, pt)}$ is a pair of a node (host or switch) $n$ and a port $pt$. We describe the network topology as a bijective function $\lambda: (\mathit{Switches} \cup \mathit{Hosts}) \times \mathit{Ports} \rightarrow (\mathit{Switches} \cup \mathit{Hosts}) \times \mathit{Ports}$ consisting of a set of directed edges $\langle \mathit{(n, pt), (n', pt')} \rangle$, where $\mathit{pt'}$ is the input port of the switch or host $n'$ that is connected to port $\mathit{pt}$ at host or switch $n$. \emph{Hosts}, \emph{Switches} and \emph{Ports} are the (finite) sets of all hosts, switches and ports in the network, respectively. The topology function is used when a packet needs to be forwarded in the network. The location of the next hop node is decided when a \emph{send}, \emph{match} or \emph{fwd} action (all defined further below) is fired. Every SDN model is w.r.t.\ a fixed topology $\lambda$ that does not change.

\noindent\textbf{Packets. }Packets are modelled as finite bit vectors and transferred in the network by being stored to the queues of the various network components. A $packet \in \mathit{Packets}$ (the set of all packets that can appear in the network) contains bits describing the proof-relevant header information and its location $loc$.

\noindent\textbf{Hosts. }Each $\mathit{host} \in \mathit{Hosts}$, has a packet queue (\emph{rcvq}) and a finite set of ports which are connected to ports of other switches. A host can send a packet to one or more switches it is connected to (\emph{send} action in Figure \ref{model}) or receive a packet from its own $rcvq$ (\emph{recv} action in Figure \ref{model}). Sending occurs repeatedly in a non-deterministic fashion which we model implicitly via the $(0,\infty)$ abstraction at switches' packet queues, as discussed further below.

\noindent\textbf{Switches. }Each $\mathit{switch} \in \mathit{Switches}$, has a flow table $\mathit{(ft)}$, a packet queue $\mathit{(pq)}$, a control queue $\mathit{(cq)}$, a forwarding queue $\mathit{(fq)}$ and one or more ports, through which it is connected to other switches and/or hosts. A flow table $\mathit{ft} \subseteq \mathit{Rules}$ is a set of forwarding rules (with $\mathit{Rules}$ being the set of all rules). Each one consists of a tuple $\mathit{(priority, pattern, ports)}$, where $\mathit{priority} \in \mathbb{N}$ determines the priority of the rule over others, $\mathit{pattern}$ is a proposition over the proof-relevant header of a packet, and $\mathit{ports}$ is a subset of the switch's ports.
 Switches match packets in their packet queues against rules (i.e. their respective $\mathit{pattern}$) in their flow table (\emph{match} action in Figure \ref{model}) and forward packets to a connected device (or final destination), accordingly. Packets that cannot be matched to any rule are sent to the controller's request queue (\emph{rq}) (\emph{nomatch} action in Figure \ref{model}); in OpenFlow, this is done by sending a \emph{PacketIn} message. The forwarding queue \emph{fq} stores packets forwarded by the controller in \emph{PacketOut} messages. The control queue stores messages sent by the controller in \emph{FlowMod} and \emph{BarrierReq} messages. \emph{FlowMod} messages contain instructions to add or delete rules from the flow table (that trigger \emph{add} and \emph{del} actions in Figure~\ref{model}). \emph{BarrierReq} messages contain barriers to synchronise the addition and removal of rules.
 \toolname\ conforms to the OpenFlow specifications and always execute instructions in an interleaved fashion obeying the ordering constraints imposed by barriers.

\noindent\textbf{OpenFlow Controller. }The controller is modelled as a finite state automaton embedded into the overall transition system. 
A controller program {\sc cp}, as used to parametrise an SDN model, consists of $(\mathit{CS},\mathit{pktIn},\mathit{barrierIn})$.
It uses its own local state $\mathit{cs}\in \mathit{CS}$, where $\mathit{CS}$ is the finite set of control program states. Incoming \emph{PacketIn} and \emph{BarrierRes} messages from the SDN model are stored in separate queues ($\mathit{rq}$ and $\mathit{brq}$, respectively) and trigger \emph{ctrl} or \emph{bsync} actions (see Figure~\ref{model}) which are then processed by the controller program in its current state. 
The controller's corresponding handler, \emph{pktIn} for \emph{PacketIn} messages and \emph{barrierIn} for \emph{BarrierRes} messages, responds by potentially changing its local state and sending messages to a subset of \emph{Switches}, as follows. A number of \emph{PacketOut} messages -- pairs of (\emph{pkt, ports}) -- can be sent to a subset of \emph{Switches}. Such a message is stored in a switch's forward queue and instructs it to forward packet \emph{pkt} along the ports \emph{ports}. The controller may also send any number of \emph{FlowMod} and \emph{BarrierReq} messages to the control queue of any subset of \emph{Switches}. A \emph{FlowMod} message may contain an \textit{add} or \textit{delete} rule modification instruction. These are executed in an arbitrary order by switches, and \textit{barriers} are used to synchronise their execution. Barriers are sent by the controller in \emph{BarrierReq} messages. OpenFlow requires that a response message (\emph{BarrierRes}) is sent to the controller by a switch when a barrier is consumed from its control queue so that the controller can synchronise subsequent actions. 
Our model includes a \emph{brepl} action that models the sending of a \emph{BarrierRes} message from a switch to the controller's barrier reply queue (\emph{brq}), and a \emph{bsync} action that enables the controller program to react to barrier responses. 

\noindent\textbf{Queues. }All queues in the network are modelled as \emph{finite} state. Packet queues $\mathit{pq}$ for switches are modelled as multisets, and we adopt $(0,\infty)$ abstraction \cite{01infty}; i.e.\ a packet is assumed to appear either zero or an arbitrary (unbounded) amount of times in the respective multiset. This means that once a packet has arrived at a switch or host, (infinitely) many other packets of the same kind repeatedly arrive at this switch or host. Switches' forwarding queues \emph{fq} are, by contrast, modelled as sets, therefore if multiple identical packets are sent by the controller to a switch, only one will be stored in the queue and eventually forwarded by the switch. The controller's request $\mathit{rq}$ and barrier reply queues \emph{brq} are modelled as sets as well.
Hosts' receive queues $\mathit{rcvq}$ are also modelled as sets. Controller queues $cq$ at switches are modelled as a finite sequence of sets of control messages (representing add and remove rule instructions), interleaved by any number of barriers. As the number of barriers that can appear at any execution is finite, this sequence is finite.

\subsection{Guarded Transitions}
\label{transitions}

Here we provide a detailed breakdown of the transition relation $ s \xhookrightarrow[]{  \alpha(\vec{a})  } s'$ for each action $\alpha(\vec{a}) \in A(s)$, where $A(s)$ the set of all enabled actions in $s$ in the proposed model (see Figure \ref{model}). Transitions are labelled by action names $\alpha$ with arguments $\vec{a}$. The transitions are only enabled in state $s$ if $s$ satisfies certain conditions called \emph{guards} that can refer to the arguments $\vec{a}$. In guards, we make use of predicate $\mathit{bestmatch(sw,r,pkt)}$ that expresses that $r$ is the highest priority rule in $\mathit{sw.ft}$ that matches $\mathit{pkt}$'s header.
Below we list all possible actions with their respective guards.

\renewcommand{\labelitemi}{$\bullet$}

	\noindent $\boldsymbol{\mathit{send(h,pt,pkt)}}$. Guard: $\mathit{true}$. This transition models packets arriving in the network in a non-deterministic fashion. When it is executed, $\mathit{pkt}$ is added to the packet queue of the network switch connected to the port $\mathit{pt}$ of host $h$ (or, formally, to $\lambda(\mathit{h,pt})_1.\mathit{pq}$, where $\lambda$ is the topology function described above). As described in \S\ref{subsec:state:repr}, only relevant representatives of packets are actually sent by end-hosts. This transition is unguarded, therefore it is always enabled.

	\noindent $\boldsymbol{\mathit{recv(h, pkt)}}$. Guard: $\mathit{pkt} \in \mathit{h.rcvq}$. This transition models hosts receiving (and removing) packets from the network and is enabled if $\mathit{pkt}$ is in $h$'s receive queue. 
	%
	%
	
	\noindent $\boldsymbol{\mathit{match(sw,pkt,r)}}$. Guard: $\mathit{pkt} \in \mathit{sw.pq} \wedge r\in \mathit{sw.ft} \wedge
\mathit{bestmatch(sw,r,pkt)}$.
	This transition models matching and forwarding packet $\mathit{pkt}$ to zero or more next hop nodes (hosts and switches), as a result of highest priority matching of rule $r$ with $\mathit{pkt}$. The packet is then copied to the packet queues of the connected hosts and/or switches, by applying the topology function to the port numbers in the matched rule; i.e. $\lambda(\mathit{sw,pt})_1.\mathit{pq}, \forall \mathit{pt} \in \mathit{r.ports}$. Dropping packets is modelled by having a special `drop' port that can be included in rules. The location of the forwarded packet(s) is updated with the respective destination (switch/host, port) pair; i.e.\ $\lambda(\mathit{sw,pt})$. Due to the $(0,\infty)$ abstraction, the packet is not removed from $\mathit{sw.pq}$.

	\noindent $\boldsymbol{\mathit{nomatch(sw,pkt)}}$. Guard: $\mathit{pkt} \in \mathit{sw.pq}  \wedge \nexists r\in \mathit{sw.ft}~.~
\mathit{bestmatch(sw,r,pkt)}$. This transition models forwarding a packet to the OpenFlow controller when a switch does not have a rule in its forwarding table that can be matched against the packet header. In this case, $\mathit{pkt}$ is added to \emph{rq} for processing. $\mathit{pkt}$ is not removed from $\mathit{sw.pq}$ due to the supported $(0,\infty)$ abstraction.

	\noindent $\boldsymbol{\mathit{ctrl(pkt, cs)}}$. Guard: $\mathit{pkt} \in \mathit{controller.rq}$. This transition models the execution of the packet handler by the controller when packet $\mathit{pkt}$, that was previously sent by switch $\mathit{pkt.loc}_1$, is available in \emph{rq}. The controller's packet handler function $\mathit{pktIn(pkt.loc}_1, \mathit{pkt, cs)}$ is executed which, in turn
	\begin{inparaenum}[(i)]
		\item reads the current controller state \emph{cs} and changes it according to the controller program,
		\item adds a number of rules, interleaved with any number of barriers, into the \emph{cq} of zero or more switches, and 
		\item adds zero or more forwarding messages, each one including a packet along with a set of ports, to the \emph{fq} of zero or more switches.
	\end{inparaenum}

	\noindent $\boldsymbol{\mathit{fwd(sw, pkt, ports)}}$. Guard: $\mathit{(pkt,ports)} \in \mathit{sw.fq}$. This transition models forwarding packet $\mathit{pkt}$ that was previously sent by the controller to $sw$'s forwarding queue $\mathit{sw.fq}$. In this case, $\mathit{pkt}$ is removed from $\mathit{sw.fq}$ (which is modelled as a set), and added to the $\mathit{pq}$ of a number of network nodes (switches and/or hosts), as defined by the topology function $\lambda(\mathit{sw,pt})_1.\mathit{pq}, \forall \mathit{pt} \in \mathit{ports}$. The location of the forwarded packet(s) is updated with the respective destination (switch/host, port) pair; i.e.\ $\lambda(\mathit{n,pt})$.

	\noindent $\boldsymbol{\mathit{FM(sw, r)}}$, where $\mathit{FM}\in \{\mathit{add, del}\}$. Guard: $(\mathit{FM,r}) \in \mathit{head(sw.cq)}$.
	 These transitions model the addition and deletion, respectively, of a rule in the flow table of switch \emph{sw}.
	 They are enabled when one or more \emph{add} and \emph{del} control messages are in the set at the head of the switch's control queue. In this case, $r$ is added to -- or deleted from, respectively -- $\mathit{sw.ft}$ and the control message is deleted from the set at the head of \emph{cq}. If the set at the head of \emph{cq} becomes empty it is removed. If then the next item in \emph{cq} is a barrier, a \emph{brepl} transition becomes enabled (see below).

	\noindent $\boldsymbol{\mathit{brepl(sw, xid)}}$. Guard: $\mathit{b(xid)} = \mathit{head(sw.cq)}$. This transition models a switch sending a barrier response message, upon consuming a barrier from the head of its control queue; i.e.\ if $\mathit{b(xid)}$ is the head of \emph{sw.cq}, where $\mathit{xid} \in \mathbb{N}$ is an identifier for the barrier set by the controller, $\mathit{b(xid)}$ is removed and the barrier reply message $\mathit{br(sw, xid)}$ is added to the controller's \emph{brq}.

	\noindent $\boldsymbol{\mathit{bsync(sw, xid, cs)}}$. Guard: $\mathit{br(sw, xid)} \in \mathit{controller.brq}$.
	 This transition models the execution of the barrier response handler by the controller when a barrier response sent by switch $sw$ is available in \emph{brq}. In this case, $\mathit{br(sw, xid)}$ is removed from the \emph{brq}, and the controller's barrier handler $\mathit{barrierIn(sw,xid,cs)}$ is executed which, in turn
		\begin{inparaenum}[(i)]
		\item reads the current controller state $cs$ and changes it according to the controller program,
		\item adds a number of rules, interleaved with any number of barriers, into the $cq$ of zero or more switches, and 
		\item adds zero or more forwarding messages, each one including a packet along with a set of ports, to the \emph{fq} of zero or more switches.
	\end{inparaenum}

\noindent\textbf{An example run}. In Figure \ref{tab:transitions}, we illustrate a sequence of MOCS transitions through a simple packet forwarding example. The run starts with a \textit{send} transition; packet \textit{p} is copied to the packet queue of the switch in black. Initially, switches' flow tables are empty, therefore \textit{p} is copied to the controller's request queue (\textit{nomatch} transition); note that \textit{p} remains in the packet queue of the switch in black due to the $(0,\infty)$ abstraction. The controller's packet handler is then called (\textit{ctrl} transition) and, as a result, (1) \textit{p} is copied to the forwarding queue of the switch in black, (2) rule $r_1$ is copied to the control queue of the switch in black, and (3) rule $r_2$ is copied to the control queue of the switch in white. Then, the switch in black forwards \textit{p} to the packet queue of the switch in white (\textit{fwd} transition). The switch in white installs $r_2$ in its flow table (\textit{add} transition) and then matches \textit{p} with the newly installed rule and forwards it to the receive queue of the host in white (\textit{match} transition), which removes it from the network (\textit{recv} transition).

\begin{figure}[h]
	\centering
	\vspace{.4cm}
	\includegraphics[width=1\linewidth]{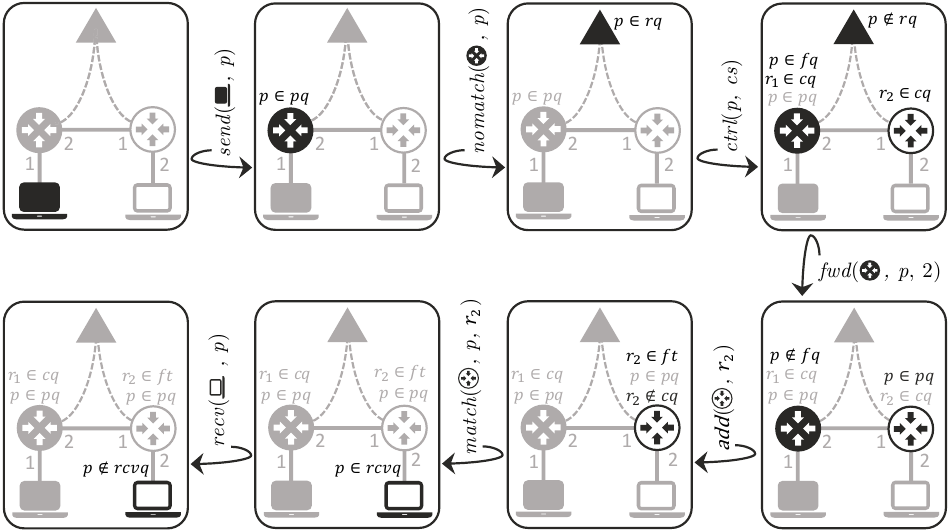}
	\caption{Forwarding $p$ from $\protect\blacklaptop$ to $\protect\whitelaptop$. Non greyed-out icons are the ones whose state changes in the current transition.}
	\label{tab:transitions}
\end{figure}

\subsection{Specification Language} 
\label{spec}

In order to specify properties of packet flow in the network, we use LTL formulas without ``next-step" operator $\bigcirc$\footnote{This is the largest set of formulae supporting the partial order reductions used in \S\ref{sec:optimisations}, as stutter equivalence does not preserve the truth value of formulae with the $\bigcirc$.}, where atomic formulae denoting properties of states of the transition system, i.e.\ SDN network.
In the case of safety properties, i.e.\ an invariant w.r.t.\ states, the LTL$_{\setminus \{\bigcirc\}}$ formula is of the form $\Box \varphi$, i.e.\ has only an outermost $\Box$ temporal connective. 

Let $P$ denote unary predicates on packets which encode a property of a packet based on its fields.
An atomic \textit{state condition} (proposition) in $\mathit{AP}$ is either of the following: (i) existence of a packet $\mathit{pkt}$ located in a packet queue ($pq$) of a switch or in a receive queue (\emph{rcvq}) of a host that satisfies $P$ (we denote this by
$\exists \mathit{pkt} {\in} \mathit{n.pq}\,.\, \mathit{P(pkt)}$ with $n\in  \mathit{Switches}$, and $\exists \mathit{pkt} {\in} \mathit{h.rcvq}\,.\, \mathit{P(pkt)}$ with $h\in \mathit{Hosts}$)\footnote{Note that these are \emph{atomic} propositions despite the use of the existential quantifier notation.};
(ii) the controller is in a specific \emph{controller} state $q\in \mathit{CS}$, denoted by a unary predicate symbol $Q(q)$ which holds in system state $s \in S$ if $q = s.\gamma.cs$.
The specification logic comprises first-order formula with equality on the finite domains of switches, hosts, rule priorities, and ports which are \emph{state-independent} (and decidable).

For example,
$\exists \mathit{pkt} {\in} \mathit{sw.pq}\,.\, \mathit{P(pkt)}$
represents the fact that the packet predicate $P(\_)$ is true for at least one packet $\mathit{pkt}$ in the \emph{pq} of switch $sw$. 
For every atomic packet proposition $\mathit{P(pkt)}$, also its negation $\neg \mathit{P(pkt)}$ is an atomic proposition for the reason of simplifying syntactic checks of formulae in Table~\ref{tab:safeness} in the next section. Note that universal quantification over packets in a queue is a derived notion. For instance, $\forall \mathit{pkt} {\in} \mathit{n.pq}\,.\, \mathit{P(pkt)}$ can be expressed as $ \nexists \mathit{pk} {\in} \mathit{n.pq}\,.\, \neg \mathit{P(pkt)}$.
Universal and existential quantification over switches or hosts can be expressed by finite iterations of ${\wedge}$ and ${\vee}$, respectively.

In order to be able to express that a condition holds when a certain event happened, we add to our propositions instances of \emph{propositional dynamic logic} \cite{DL1,DL2}.
Given an action $\alpha(\cdot)\in A$ and a proposition $P$ that may refer to any variables in $\vec{x}$, $[\alpha(\vec{x})]P$ is also a proposition and $[\alpha(\vec{x})]P$ is true if, and only if,
after firing transition $\alpha(\vec{a})$ (to get to the current state), $P$ holds with the variables in $\vec{x}$ bound to the corresponding values in the actual arguments $\vec{a}$. 
With the help of those basic modalities one can then also specify that more complex events occurred. For instance, dropping of a packet due to a \emph{match} or \emph{fwd} action can be expressed by $[\mathit{match(sw,pkt,r)}](r.\mathit{fwd\_port}=\texttt{drop}) \land [\mathit{fwd(sw,pkt,pt)}](\mathit{pt}=\texttt{drop})$. Such predicates derived from modalities are used in \S \ref{appendix:CP}-\ref{rq1-alg3}.


The meaning of temporal LTL operators is standard depending on the trace of a transition sequence $s_0 \xhookrightarrow{\alpha_1} s_1 \xhookrightarrow{\alpha_2} \ldots$.
The trace $L(s_0)L(s_1)\ldots L(s_i)\ldots$ is defined as usual. 
For instance, trace $L(s_0)L(s_1)L(s_2)\ldots$ satisfies invariant $\Box \varphi$ if each $L(s_i)$ implies $\varphi$. 

\section{Model Checking}
\label{sec:optimisations}

In order to verify desired properties of an SDN, we use its model as described in Def.~\ref{def:SDNmodel}
and apply model checking. In the following we propose optimisations that significantly improve the performance of model checking. 

\subsection{Contextual Partial-Order Reduction}
\label{subsec:por}

\emph{Partial order reduction} (POR) \cite{representatives} reduces the number of interleavings (traces) one has to check. Here is a reminder of the main result (see \cite{Baier2008}) where we use a stronger condition than the regular (\emph{C4}) to deal with cycles:

\begin{thm}[Correctness of POR]
Given a finite transition system $\mathcal{M} = (S, A, \hookrightarrow, s_0, AP, L)$ that is action-deterministic and without terminal states, let $A(s)$ denote the set of actions in $A$ enabled in state $s\in S$.
Let $\mathit{ample}(s) \subseteq A(s)$ be a set of actions for a state $s\in S$ that satisfies the following conditions:
\begin{enumerate}[label=C\arabic*]
\item (Non)emptiness condition: $\varnothing  \neq \mathit{ample}(s) \subseteq A(s)$.
\item Dependency condition: Let $s \xhookrightarrow[]{\alpha_1} s_1...\xhookrightarrow[]{\alpha_n} s_n \xhookrightarrow[]{\beta} t$ be a run in $\mathcal{M}$. If $\beta \in A  \setminus \mathit{ample}(s)$ depends on $\mathit{ample}(s)$, then $\alpha_i \in \mathit{ample}(s)$ for some $0 < i \leq n$, which means that in every path fragment of $\mathcal{M} $, $\beta$ cannot appear before some transition from $\mathit{ample}(s)$ is executed.
\item Invisibility condition: If $\mathit{ample}(s) \neq A(s)$ (i.e., state $s$ is not fully expanded), then every $\alpha \in \mathit{ample}(s)$ is invisible.
\item Every cycle in $\mathcal{M}^{\mathit{ample}}$ contains a state $s$ such that $\mathit{ample}(s)=A(s)$.
\end{enumerate}
\label{ample}
where $\mathit{\mathcal{M}^{\mathit{ample}} = (S_a, A, \hookrightarrowdbl, s_0, AP, L_a) }$ is the new, optimised, model defined as follows:
let 
$S_a \subseteq S$ be the set of states reachable from the initial state $s_0$ under $\hookrightarrowdbl$, let $L_a(s) = L(s)$ for all $s \in S_a$ and define $\hookrightarrowdbl\,    \subseteq S_a\times A\times S_a$ inductively by the rule
\[
\frac {s~ \xhookrightarrow[]{ \alpha}    ~s'}{ s~ \stackrel{~\alpha~}{\hookrightarrowdbl}    s'  }
\qquad   \textup{if }\  \alpha \in \mathit{ample}(s)  \]
If $\mathit{ample}(s)$ satisfies conditions (C1)-(C4) as outlined above, then for each path in $\mathcal{M} $ there exists a stutter-trace equivalent path in $\mathcal{M}^{\mathit{ample}} $, and vice-versa, denoted $\mathcal{M}   \stutt \mathcal{M}^{\mathit{ample}} $. 	\label{thm1}
\end{thm}
%
The intuitive reason for this theorem to hold is the following: Assume an action sequence $\alpha_i . . . \alpha_{i+n}\beta$ that reaches the state $s$, and $\beta$ is \textit{independent} of $\{\alpha_i, . . . \alpha_{i+n}\}$. Then, one can permute $\beta$ with $\alpha_{i+n}$ through $\alpha_i$ successively $n$ times. One can therefore construct the sequence $\beta\alpha_i . . . \alpha_{i+n}$ that also reaches the state $s$. If this shift of $\beta$ does not affect the labelling of the states with atomic propositions ($\beta$ is called \emph{invisible} in this case), then it is not detectable by the property to be shown and the permuted and the original sequence are equivalent w.r.t.\ the property and thus don't have to be checked both. 
One must, however, ensure, that in case of loops (infinite execution traces) the ample sets do not \emph{preclude} some actions to be fired altogether, which is why one needs (\emph{C4}).

%
%

The more actions that are both stutter and provably independent (also referred to as \textit{safe actions} \cite{Peled95}) there are, the smaller the transition system, and the more efficient the model checking.
One of our contributions is that we attempt to identify \emph{as many safe actions as possible} to make PORs more widely applicable to our model.

The PORs in \cite{Kuai} consider only dependency and invisibility of \emph{recv} and \emph{barrier} actions, whereas we explore systematically all possibilities for applications of Theorem~\ref{thm1} to reduce the search space.
When identifying safe actions, we consider (1) the actual controller program {\sc cp}, (2) the topology $\lambda$ and (3) the state formula $\varphi$ to be shown invariant, which we call the \emph{context} {\sc ctx} of actions. 
It turns out that two actions may be dependent in a given context of abstraction while independent in another context, and similarly for invisibility, and we exploit this fact. The argument of the action thus becomes relevant as well.

\begin{defn}[Safe Actions] \label{def:safe} Given a context {\sc ctx} $= (\textsc{cp},\lambda,\varphi)$, and SDN model $\mathcal{M}_{(\lambda,\textsc{cp})} = (S, A, \hookrightarrow, s_0, AP, L)$, an action $\alpha(\cdot)\in A(s)$ is called `safe' if it is independent of any other action in $A$ and invisible for $\varphi$.
We write safe actions $\check{\alpha}(\cdot)$.
\end{defn}
%
\begin{defn}[Order-sensitive Controller Program]
	\label{def:ord-sens}
	A controller program {\sc cp} is order-sensitive if there exists a state $s\in S$ and two actions $\alpha,~\beta$ in $\{ \mathit{ctrl}(\cdot), \mathit{bsync}(\cdot)  \}$ such that $\alpha,\beta \in A(s)$
and $s\xhookrightarrow[]{\alpha} s_1 \xhookrightarrow[]{\beta} s_2$ and $s\xhookrightarrow[]{\beta} s_3 \xhookrightarrow[]{\alpha} s_4$ with $s_2\neq s_4$.
\end{defn}

\begin{defn} \label{def:unobservableAction}
Let $\varphi$ be a state formula. An action $\alpha \in A$ is called `$\varphi$-invariant' if 
$s\models \varphi$ iff $\alpha(s)\models \varphi$ for all $s\in S$ with $\alpha\in A(s)$.
\end{defn}


\begin{lemma}\label{lemma:safe} For transition system $\mathcal{M}_{(\lambda,\textsc{cp})} = (S, A, \hookrightarrow, s_0, AP, L)$ and a formula $\varphi\in \text{LTL}_{ \setminus \{\bigcirc \}}$, 
$\alpha \in A  \text{ is safe   }~\mathrm{iff}~  \bigwedge\nolimits^3_{i=1} \mathit{Safe}_i(\alpha)$,  
	where $\mathit{Safe}_i$, given in Table \ref{tab:safeness}, are per-row. 
	%
\end{lemma}
\longversion{
\begin{proof} 
\vspace{-0.2cm}	See Appendix \ref{appendix:proofs}.
\end{proof}}
\begin{table}[h]
	\scriptsize
	\begin{center}
		\caption{Safeness Predicates}
		\label{tab:safeness}
		\begin{tabular}{m{2.9cm}|m{4.4cm}|m{4.8cm}} 
			\toprule 
			\multicolumn{1}{c|}{\textbf{Action}} & \multicolumn{1}{c|}{\textbf{Independence} }& \multicolumn{1}{c}{\textbf{Invisibility}}\\
			\multicolumn{1}{c|}{$\mathit{Safe}_1(\alpha)$} & \multicolumn{1}{c|}{$\mathit{Safe}_2(\alpha)$ } & \multicolumn{1}{c}{$\mathit{Safe}_3(\alpha)$} \\
			\midrule 
			$\alpha =\mathit{ctrl(pk, cs)}$ &  {\sc cp} is not order-sensitive 
			& \texttt{if} $Q(q)$ occurs in $\varphi$, where $q \in \mathit{CS}$, \texttt{then}   $\alpha$ is $\varphi$-invariant.\\
			\midrule
			$\alpha = \mathit{bsync(sw, xid, cs)}$ &   {\sc cp} is not order-sensitive   &  \texttt{if}  $Q(q)$ 
			occurs in $\varphi$, where $q \in \mathit{CS}$, \texttt{then}  $\alpha$ is $\varphi$-invariant.\\
			\midrule
			$\mathit{\alpha = fwd(sw, pk, \mathit{ports})}$ & \multicolumn{1}{c|}{$\top$} &  \texttt{if}  $\exists pk {\in} b.q\,.\, P(pk)$ occurs in $\varphi$, for any $b\in \{sw\}\cup \{ \lambda(sw,p)_1\ |\ p\in \mathit{ports} \}$ and $q \in \{\mathit{pq,recvq}\}$, \texttt{then}    $\alpha$ is $\varphi$-invariant.\\
			\midrule
			$\alpha =  \mathit{brepl(sw, xid)}$ & \multicolumn{1}{c|}{$\top$} & \multicolumn{1}{c}{$\top$}\\
			\midrule
			$\alpha = \mathit{recv(h, pk)}$ & \multicolumn{1}{c|}{$\top$} &  \texttt{if}  $\exists pk {\in} \mathit{h.rcvq}\,.\, \mathit{P(pk)}$ occurs in $\varphi$, \texttt{then}   $\alpha$ is  $\varphi$-invariant.\\
			\bottomrule 
		\end{tabular}
	\end{center}
	\vspace{-.3cm}
\end{table}

%
%




\begin{thm}[POR instance for SDN]
\label{thm:POR}
Let $(\textsc{cp}, \lambda, \varphi)$ be a context such that $\mathcal{M}_{(\lambda,\textsc{cp})} = (S, A, \hookrightarrow, s_0, AP, L)$ is an SDN network model from Def.~\ref{def:SDNmodel};
and let safe actions be as in Def.~\ref{def:safe}. Further, let $\mathit{ample}(s)$ be defined by:
\[ 
\mathit{ample}(s)= \left\{ \begin{array}{ll}
\{\alpha\in A(s)\ |\ \alpha \text{ safe } \} & \text{if } \{\alpha\in A(s)\ |\ \alpha \text{ safe } \} \neq \emptyset\\
A(s) &\text{otherwise}
\end{array}
\right.
\]
Then, $\mathit{ample}$ satisfies the criteria of Theorem~\ref{ample} and thus $\mathcal{M}_{(\lambda,\textsc{cp})}       \stutt      \mathcal{M}^{\mathit{ample}}_{(\lambda,\textsc{cp})}$ \footnote{Stutter equivalence here implicitly is defined w.r.t.\  the atomic propositions appearing in $\varphi$,  but this suffices as we are just interested in the validity of $\varphi$.}
\end{thm}
\begin{proof}
	$ $\newline
	\vspace{-.4cm}
	\begin{enumerate}[label=\textit{C\arabic*}]
		\item The (non)emptiness condition is trivial since by definition of $\mathit{ample}(s)$ it follows that $\mathit{ample}(s) = \varnothing $ iff $A(s) = \varnothing $.
		\item By assumption $\beta \in A  \setminus \mathit{ample}(s)$ depends on $\mathit{ample}(s)$. But with our definition of $\mathit{ample}(s)$ this is impossible as all actions in
		 $\mathit{ample}(s)$ are safe and by definition independent of all other actions.
		
		%
		\item The validity of the invisibility condition is by definition of $\mathit{ample}$ and safe actions.
		\item 
		We now show that every cycle in $\mathcal{M}_{(\lambda,\textsc{cp})}^{\mathit{ample}}$ contains a fully expanded state $s$, i.e.\  a state $s$ such that $\mathit{ample}(s)=A(s)$.	
		\label{cycle}
 By definition of $\mathit{ample}(s)$ in Thm.~\ref{thm:POR} it is equivalent to show that there is no cycle in $\mathcal{M}_{(\lambda,\textsc{cp})}^{\mathit{ample}}$ consisting of safe actions only.
		We show this by contradiction, assuming such a cycle of only safe actions exists.
		There are five safe action types to consider: \emph{ctrl}, \emph{fwd}, \emph{brepl}, \emph{bsync} and \emph{recv}. 
		Distinguish two cases.
		
	 	 \emph{Case 1. A sequence of safe actions of same type}. Let us consider the different safe actions:
     		\begin{itemize}
			 \item Let $\rho$ an execution of $\mathcal{M}_{(\lambda,\textsc{cp})}^{\mathit{ample}}$ which consists of only one type of \emph{ctrl}-actions: 
				$$\rho = s_1 \xhookrightarrowdbl{\mathit{ctrl}(\mathit{pkt}_1,\mathit{cs}_1)} s_2 \xhookrightarrowdbl{\mathit{ctrl}(\mathit{pkt}_2,\mathit{cs}_2)} ... s_{i-1} \xhookrightarrowdbl{\mathit{ctrl}(\mathit{pkt}_{i-1},\mathit{cs}_{i-1})}  s_i$$
				Suppose $\rho$ is a cycle. 
				According to the \emph{ctrl} semantics, for each transition $s \xhookrightarrowdbl{\mathit{ctrl(pkt,cs)}} s' $, 
				where $s = (\pi, \delta, \gamma)$, $s' = (\pi', \delta', \gamma')$, it holds that $\mathit{\gamma'.rq} = \mathit{\gamma.rq}\setminus \{\mathit{pkt}\}$ as we use sets to represent \emph{rq} buffers. 
				Hence, for the execution $\rho$ it holds $\gamma_i.\mathit{rq} = \gamma_1.\mathit{rq}\setminus \{\mathit{pkt}_1, \mathit{pkt}_2,...\mathit{pkt}_{i-1}\}$ which implies that $s_1 \neq s_i$. Contradiction. 
			\item 		 Let $\rho$ an execution which consists of only one type of \emph{fwd}-actions: similar argument as above since \emph{fq}-s are represented by sets and thus forward messages are removed from \emph{fq}.
			\item 		 Let $\rho$ an execution which consists of only one type of \emph{brepl}-actions: similar argument as above since control messages are removed from \emph{cq}.

			\item 		 Let $\rho$ an execution which consists of only one type of \emph{bsync}-actions: similar argument as above, as barrier reply messages are removed from \emph{brq}-s  that are represented by sets.

			\item 		 Let $\rho$ an execution which consists of only one type of \emph{recv}-actions: similar argument as above, as packets are removed from $\mathit{rcvq}$ buffers that are represented by sets.
      		   \end{itemize}
		
	 	  \emph{Case 2. A sequence of different safe actions}. 
		Suppose there exists a cycle with mixed safe actions starting in $s_1$ and ending in $s_i$. 
		Distinguish the following cases.
		
		\begin{longenum}
			\item There exists at least a \emph{ctrl} and/or a \emph{bsync} action in the cycle. According to the effects of safe transitions, the \emph{ctrl} action will change to a state with smaller \emph{rq} and the \emph{bsync} will always switch to a state with smaller \emph{brq}. It is important here that \emph{ctrl} does not interfere with \emph{bsync} regarding \emph{rq, brq}, and no safe action of other type than \emph{ctrl} and \emph{bsync} accesses \emph{rq} or \emph{brq}. This implies that $s_1 \neq s_i$. Contradiction.
			\item Neither \emph{ctrl}, nor \emph{bsync} actions in the cycle.
			\begin{longenum}
				\item There is a \emph{fwd} and/or \emph{brepl} in the cycle: \emph{fwd} will always switch to a state with smaller \emph{fq} and \emph{brepl} will always switch to a state with smaller \emph{cq} (\emph{brepl} and \emph{recv} do not interfere with \emph{fwd}). This implies that $s_1 \neq s_i$. Contradiction.
				\item There is neither \emph{fwd} nor \emph{brepl} in the cycle. This means that only \emph{recv} is in the cycle which is already covered by the first case.
			\end{longenum}											
	       \end{longenum}
 	\end{enumerate}
\end{proof}

\noindent
Due to the definition of the transition system via ample sets, each safe action is immediately executed after its enabling one. Therefore, one can merge every transition of a safe action with its precursory enabling one. 
Intuitively, the semantics of 
the merged action is defined as the successive execution of 
its constituent actions. This process can be repeated if there is a chain of safe actions; for instance, in the case of ${s \xhookrightarrowdbl[]{\mathit{nomatch(sw,pkt)}}s' \xhookrightarrowdbl[]{\mathit{ctrl(pkt,cs)}} s'' \xhookrightarrowdbl[]{\mathit{fwd(sw,pkt,ports)}} s'''}$
where each transition enables the next and the last two are assumed to be safe. These transitions can be merged into one, yielding a stutter equivalent trace as the intermediate states are invisible (w.r.t.\ the context and thus the property to be shown) by definition of safe actions.

\subsection{State Representation}
\label{subsec:state:repr}

Efficient state representation is crucial for minimising \toolname's memory footprint and enabling it to scale up to relatively large network setups. 

\noindent\textbf{Packet and Rule Indexing. }In \toolname, only a single instance of each packet and rule that can appear in the modelled network is kept in memory. An index is then used to associate queues and flow tables with packets and rules, with a single bit indicating their presence (or absence). This data structure is illustrated in Figure \ref{fig:ds-layout}.
 For a data packet, a value of $1$ in the $pq$ section of the entry indicates that infinite copies of it are stored in the packet queue of the respective switch. A value of $1$ in the \emph{fq} section indicates that a single copy of the packet is stored in the forward queue of the respective switch. A value of $1$ in the $\mathit{rq}$ section indicates that a copy of the packet sent by the respective switch (when a \emph{nomatch} transition is fired) is stored in the controller's request queue. For a rule, a value of $1$ in the \emph{ft} section indicates that the rule is installed in the respective switch's flow table. A value of $1$ in the $\mathit{cq}$ section indicates that the rule is part of a \emph{FlowMod} message in the respective switch's control queue.

\begin{figure}[h] 
	\centering
	\captionsetup{width=1\textwidth}
		\vspace{-.3cm}
	\includegraphics[width=1\linewidth]{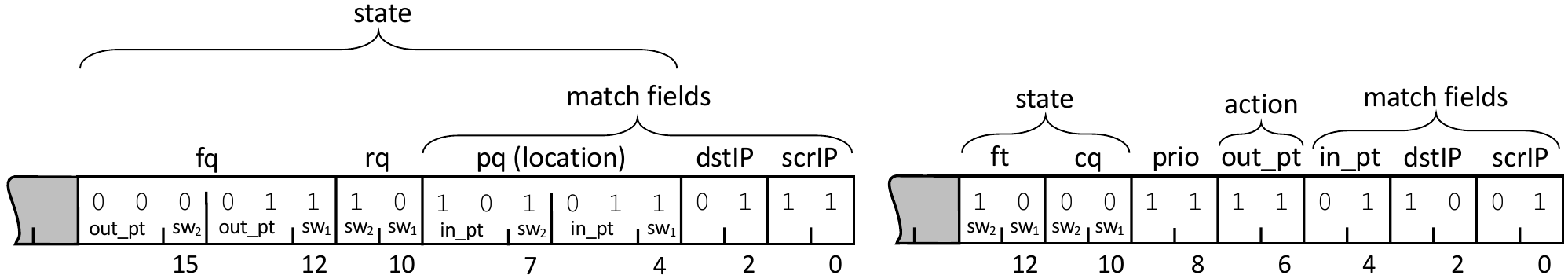}  
	\caption{Packet (left) and rule (right) indices}
	\label{fig:ds-layout}
	\vspace{-.3cm}
\end{figure}

The proposed optimisation enables scaling up the network topology by minimising the required memory footprint. For every switch, \toolname\ only requires a few bits in each packet and rule entry in the index. 

\noindent\textbf{Discovering equivalence classes of packets. }Model checking with all possible packets including all specified fields in the OpenFlow standard would entail a huge state space that would render any approach unusable. Here, we propose the discovery of equivalence classes of packets that are then used for model checking. We first remove all fields that are not referenced in a statement or rule creation or deletion in the controller program. Then, we identify packet classes that would result in the same controller behaviour. Currently, as with the rest of literature, we focus on simple controller programs where such equivalence classes can be easily identified by analysing static constraints and rule manipulation in the controller program. We then generate one representative packet from each class and assign it to all network switches that are directly connected to end-hosts; i.e. modelling clients that can send an arbitrarily large number of packets in a non-deterministic fashion. We use the minimum possible number of bits to represent the identified equivalence classes. For example, if the controller program exerts different behaviour if the destination {\sc tcp} port of a packet is $22$ (i.e.\ destined to an {\sc ssh} server) or not, we only use a 1-bit field to model this behaviour. 

\noindent\textbf{Bit packing. }We reduce the size of each recorded state by employing bit packing using the \texttt{int32} type supported by {\sc Uppaal}, and bit-level operations for the entries in the packet and rule indices, as well as for the packets and rules themselves.

\section{Experimental Evaluation}
\label{experimental-evaluation}

In this section, we experimentally evaluate \toolname\ by comparing it with the state of the art, in terms of performance (verification throughput and memory footprint) and model expressivity. We have implemented \toolname\ in {\sc Uppaal} \cite{UPPAAL} as a network of parallel automata for the controller and network switches, which communicate asynchronously by writing/reading packets to/from queues that are part of the model discussed in \S \ref{sec:modelSDN}. As discussed in \S \ref{sec:optimisations}, this is implemented by directly manipulating the packet and rule indices.

Throughout this section we will be using three examples of network controllers: (1) A \emph{stateless firewall} (\S \ref{appendix:CP}-\ref{rq1-alg11}) requires the controller to install rules to network switches that enable them to decide whether to forward a packet towards its destination or not; this is done in a stateless fashion, i.e. without having to consider any previously seen packets. For example, a controller could configure switches to block all packets whose destination {\sc tcp}  port is {\sc ssh}. (2) A \emph{stateful firewall} (\S \ref{appendix:CP}-\ref{rq1-alg12}) is similar to the stateless one but decisions can take into account previously seen packets. A classic example of this is to allow bi-directional communication between two end-hosts, when one host opens a {\sc tcp} connection to the other. Then, traffic flowing from the other host back to the connection initiator should be allowed to go through the switches on the reverse path. (3) A \emph{MAC learning application} (\S \ref{appendix:CP}-\ref{rq1-alg13}) enables the controller and switches to learn how to forward packets to their destinations (identified with respective MAC addresses). A switch sends a \emph{PacketIn} message to the controller when it receives a packet that it does not know how to forward. By looking at this packet, the controller learns a mapping of a source switch (or host) to a port of the requesting switch. It then installs a rule (by sending a \emph{FlowMod} message) that will allow that switch to forward packets back to the source switch (or host), and asks the requesting switch (by sending a \emph{PacketOut} message) to flood the packet to all its ports except the one it received the packet from. This way, the controller eventually learns all mappings, and network switches receive rules that enable them to forward traffic to their neighbours for all destinations in the network.\\

\subsection{Performance Comparison}
\label{performance}

We measure \toolname's performance, and also compare it against Kuai \cite{Kuai}\footnote{Note that parts of Kuai's source code are not publicly available, therefore we implemented it's model in {\sc Uppaal}.} using the examples described above, and we investigate the behaviour of \toolname\ as we scale up the network (switches and clients/servers). We report three metrics: (1) \emph{verification throughput} in visited states per second, (2) number of visited states, and (3) required memory. We have run all verification experiments on an 18-Core iMac pro, 2.3GHz Intel Xeon W with 128GB DDR4 memory. 

\noindent\textbf{Verification throughput. }We measure the verification throughput when running a single experiment at a time on one {\sc cpu} core and report the average and standard deviation for the first 30 minutes of each run.
In order to assess how \toolname's different optimisations affect its performance, we report results for the following system variants: (1) \toolname, (2) \toolname\ without POR, (3) \toolname\ without any optimisations (neither POR, state representation), and (4) Kuai. Figure~\ref{fig:throughput} shows the measured throughput (with error bars denoting standard deviation).

\begin{figure}[h] 
	\captionsetup{width=.9\textwidth}
	\begin{subfigure}[b]{1\textwidth}
		\captionsetup{width=.8\textwidth}
		\centering
		\includegraphics[width=1\linewidth]{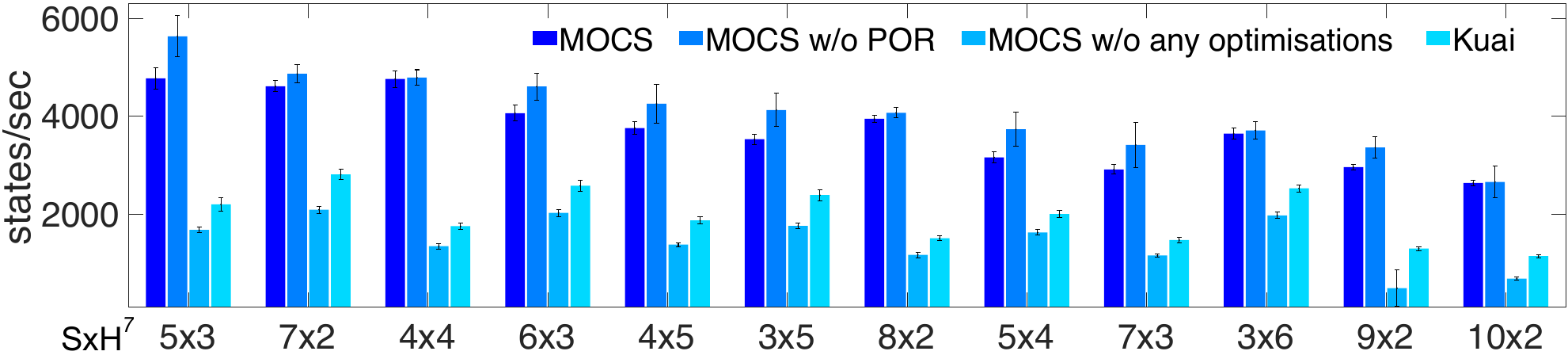}  
		\caption{MAC Learning Switch}
		\label{fig:thpML}
	\end{subfigure}\\\\
	%
	%
	\begin{subfigure}[b]{.61\textwidth}
		\captionsetup{width=1\textwidth}
		\centering
		\includegraphics[width=1\linewidth]{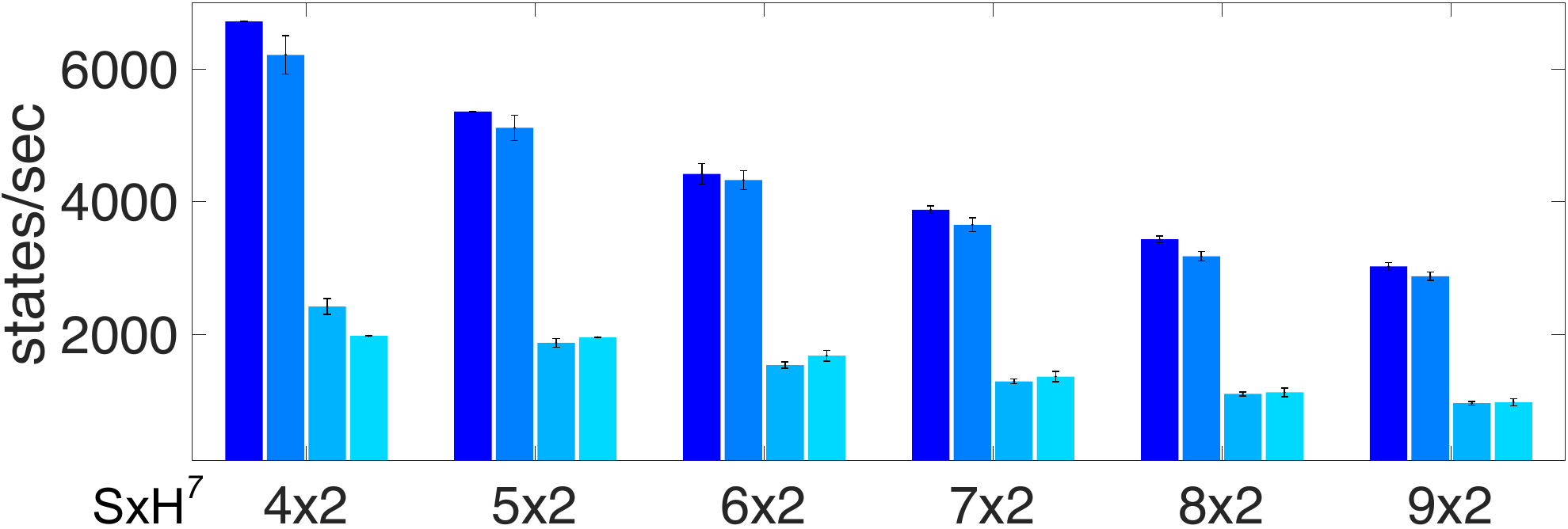}  
		\caption{Stateless Firewall}
		\label{fig:thpSSH}
	\end{subfigure}%
	\begin{subfigure}[b]{.39\textwidth}
		\captionsetup{width=.9\textwidth}
		\centering
		\includegraphics[width=1\linewidth]{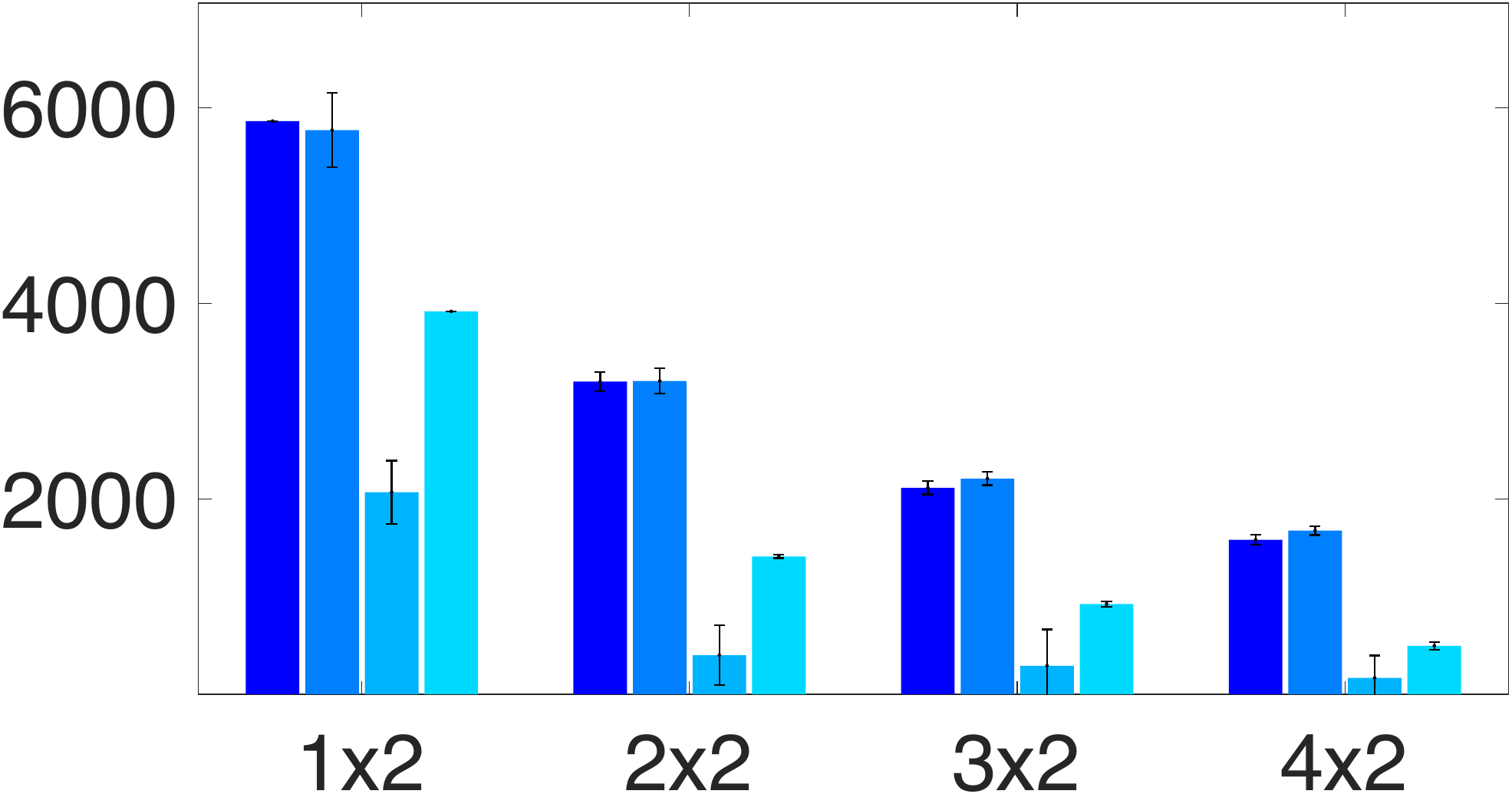}  
		\caption{Stateful Firewall}
		\label{fig:thpFW}
	\end{subfigure}%
	\caption{Performance Comparison -- Verification Throughput}
	\label{fig:throughput}
	\vspace{-.3cm}
\end{figure}

For the MAC learning and stateless firewall applications, we observe that \toolname\ performs significantly better than Kuai for all different network setups and sizes\footnote{\textsf{S}$\times$\textsf{H} in Figures \ref{fig:throughput} to \ref{fig:memory} indicates the number of switches \textsf{S} and hosts \textsf{H}.}, achieving at least double the throughput Kuai does. The throughput performance is much better for the stateful firewall, too. This is despite the fact that, for this application, Kuai employs the unrealistic optimisation where the \emph{barrier} transition forces the immediate update of the forwarding state. In other words, \toolname\ is able to explore significantly more states and identify bugs that Kuai cannot (see \S \ref{subsec:expressivity}).

The computational overhead induced by our proposed PORs is minimal. This overhead occurs when PORs require dynamic checks through the safety predicates described in Table \ref{tab:safeness}. This is shown in Figure~\ref{fig:thpML}, where, in order to decide about the (in)visibility of $\mathit{fwd(sw,pk,pt)}$ actions, a lookup is performed in the history-array of packet \emph{pk}, checking whether the bit which corresponds to switch $sw'$, which is connected with port \emph{pt} of \emph{sw}, is set. On the other hand, if a POR does not require any dynamic checks, no penalty is induced, as shown in Figures~\ref{fig:thpSSH} and \ref{fig:thpFW}, where the throughput when the PORs are disabled is almost identical to the case where PORs are enabled. This is because it has been statically established at a pre-analysis stage that all actions of a particular type are always safe for any argument/state. It is important to note that even when computational overhead is induced, PORs enable \toolname\ to scale up to larger networks because the number of visited states can be significantly reduced, as discussed below.

In order to assess the contribution of the state representation optimisation in \toolname's performance, we measure the throughput when both PORs and state representation optimisations are disabled. It is clear that they contribute significantly to the overall throughput; without these the measured throughput was at least less than half the throughput when they were enabled.

\begin{figure}[h]
	\captionsetup{width=1\textwidth}
	\begin{subfigure}[b]{1\textwidth}
		\captionsetup{width=.8\textwidth}
		\centering
		\includegraphics[width=1\linewidth]{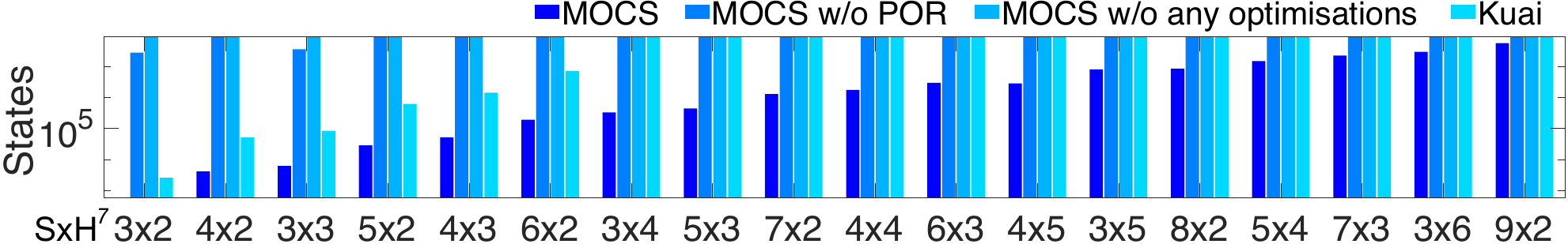}  
		\caption{MAC Learning Switch}
		\label{fig:statesML}
	\end{subfigure}\\\\
	%
	%
	\begin{subfigure}[b]{.75\textwidth}
		\captionsetup{width=1\textwidth}
		\centering
		\includegraphics[width=1\linewidth]{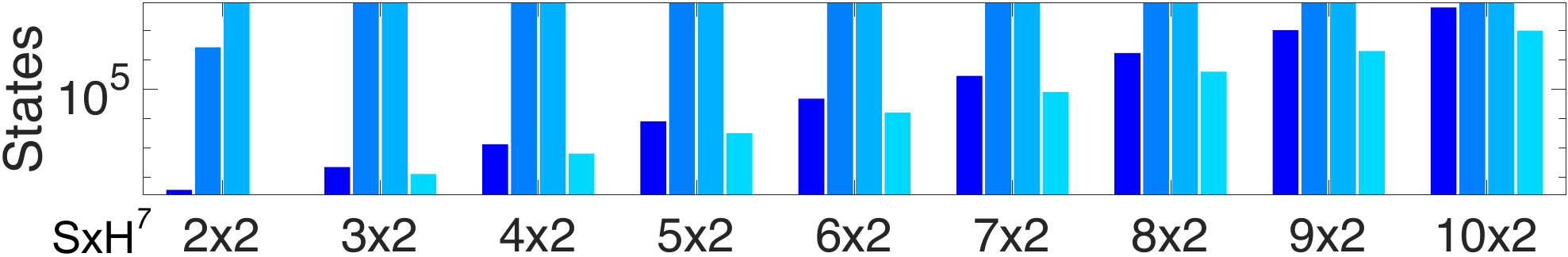}  
		\caption{Stateless Firewall}
		\label{fig:statesSSH}
	\end{subfigure}%
	\begin{subfigure}[b]{.25\textwidth}
		\captionsetup{width=1\textwidth}
		\centering
		\includegraphics[width=1\linewidth]{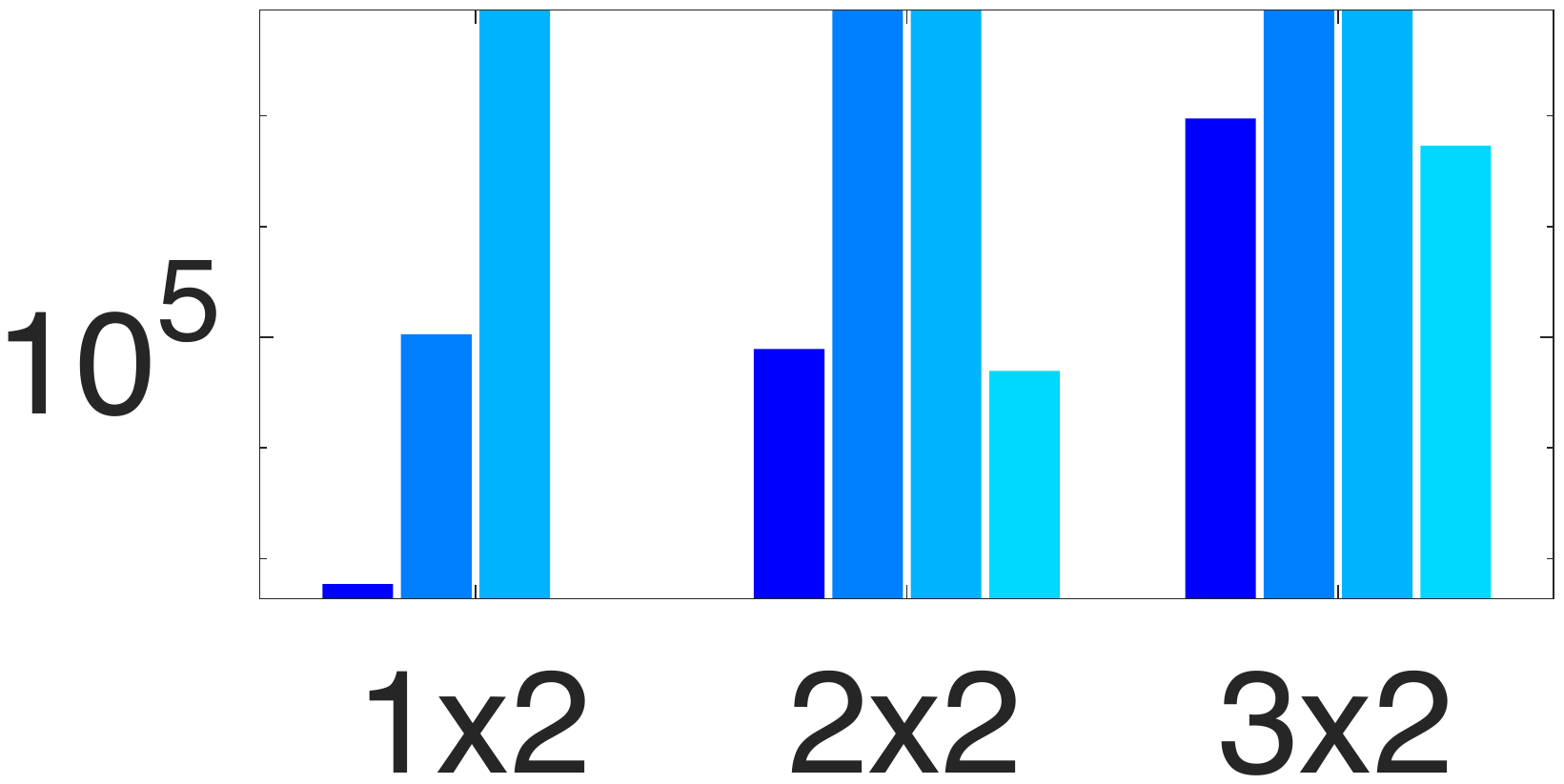}  
		\caption{Stateful Firewall}
		\label{fig:statesFW}
	\end{subfigure}%
	\caption{Performance Comparison -- Visited States (logarithmic scale)}
	\label{fig:states}
	\vspace{.4cm}
\end{figure}

\begin{figure}[h]
	\captionsetup{width=1\textwidth}
	\begin{subfigure}[b]{1\textwidth}
		\captionsetup{width=.8\textwidth}
		\centering
		\includegraphics[width=1\linewidth]{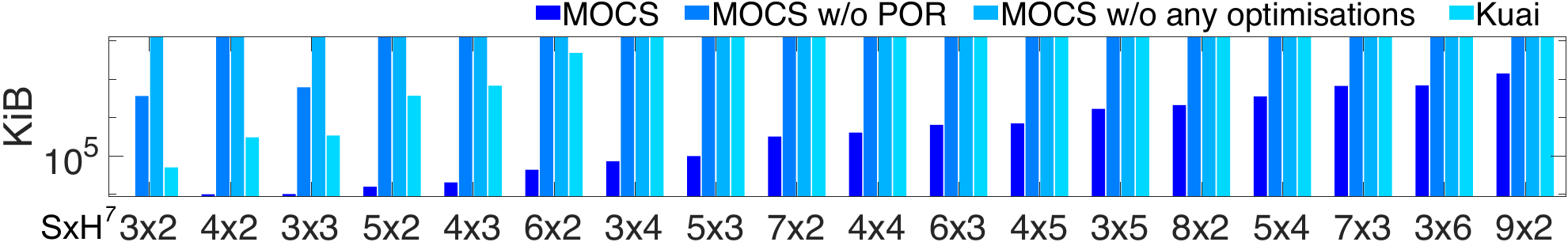}  
		\caption{MAC Learning Switch}
		\label{fig:memML}
	\end{subfigure}\\\\
	%
	%
	\begin{subfigure}[b]{.75\textwidth}
		\captionsetup{width=1\textwidth}
		\centering
		\includegraphics[width=1\linewidth]{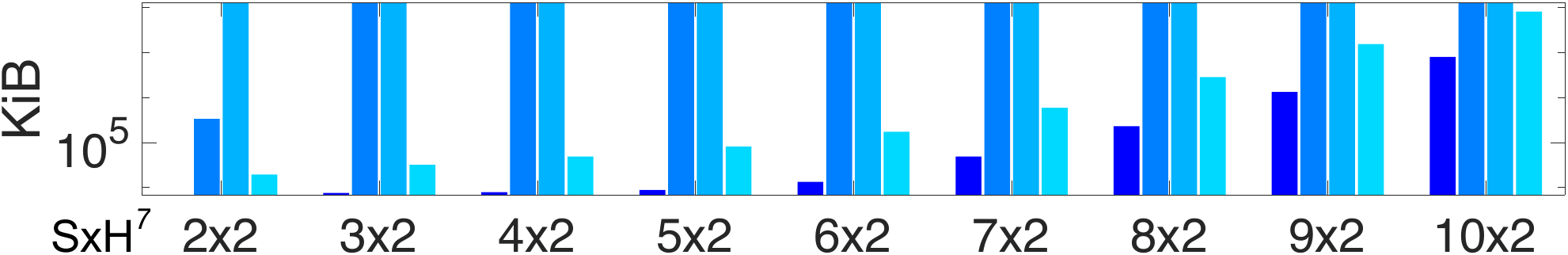}  
		\caption{Stateless Firewall}
		\label{fig:memSSH}
	\end{subfigure}%
	\begin{subfigure}[b]{.25\textwidth}
		\captionsetup{width=1\textwidth}
		\centering
		\includegraphics[width=1\linewidth]{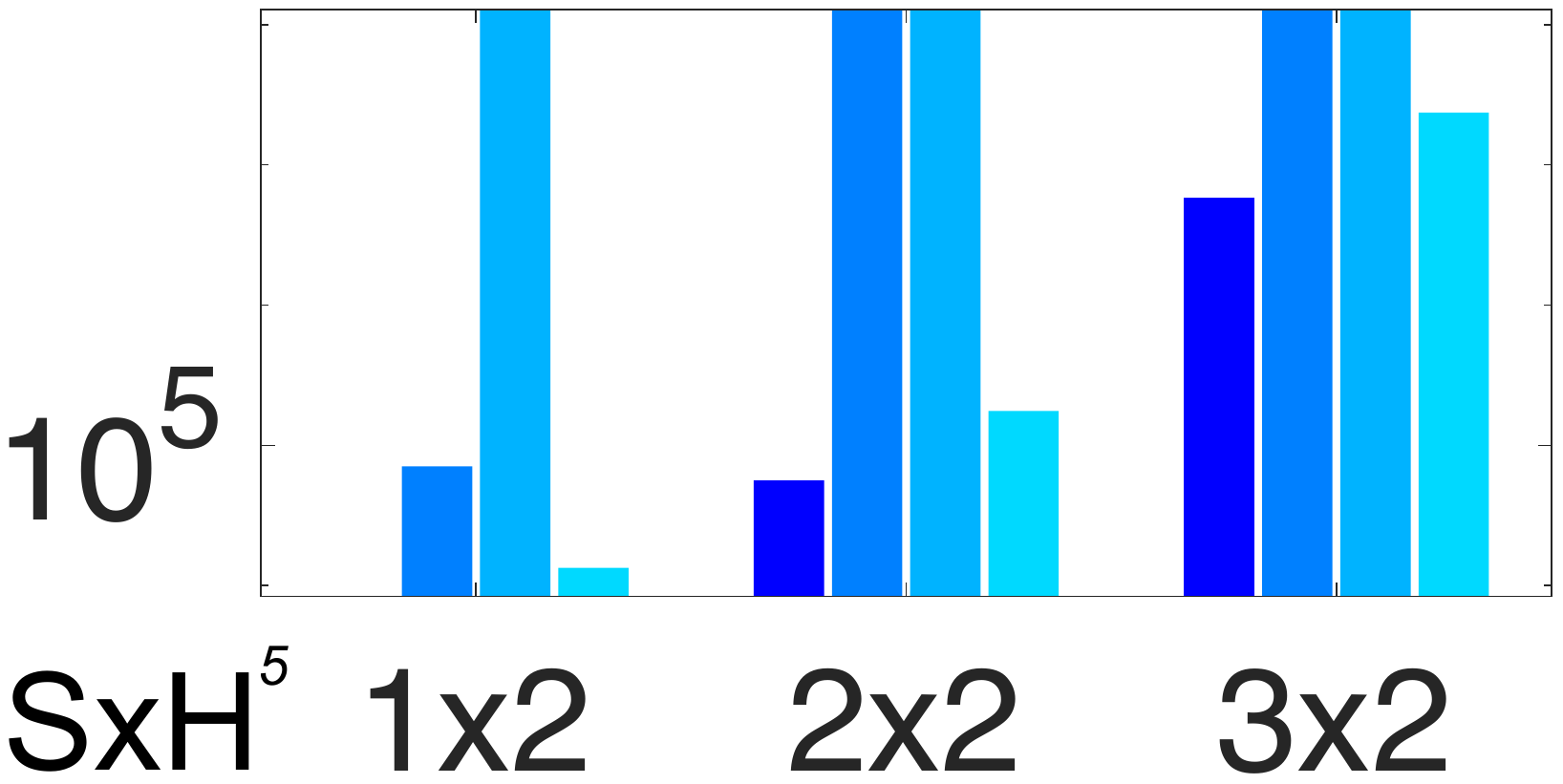}  
		\caption{Stateful Firewall}
		\label{fig:memFW}
	\end{subfigure}%
	\caption{Performance Comparison -- Memory Footprint (logarithmic scale)}
	\label{fig:memory}
	\vspace{-.3cm}
\end{figure}

\noindent\textbf{Number of visited states and required memory. }Minimising the number of visited states and required memory is crucial for scaling up verification to larger networks. The proposed partial order reductions (\S \ref{subsec:por}) and identification of packet equivalent classes aim at the former, while packet/rule indexing and bit packing aim at the latter (\S \ref{subsec:state:repr}). In Figure \ref{fig:states}, we present the results for the various setups and network deployments discussed above. We stopped scaling up the network deployment for each setup when the verification process required more than $24$ hours or started swapping memory to disk. For these cases we killed the process and report a topped-up bar in Figures \ref{fig:states} and \ref{fig:memory}. 

For the MAC learning application, \toolname\ can scale up to larger network deployments compared to Kuai, which could not verify networks consisting of more than $2$ hosts and $6$ switches. For that network deployment, Kuai visited ${\sim}7$m states, whereas \toolname\ visited only ${\sim}193$k states. At the same time, Kuai required around $48$GBs of memory ($7061$ bytes/state) whereas \toolname\ needed ${\sim}43$MBs ($228$ bytes/state). Without the partial order reductions, \toolname\ can only verify tiny networks. The contribution of the proposed state representation optimisations is also crucial; in our experiments (results not shown due to lack of space), for the $6\times 2$ network setups (the largest we could do without these optimisations), we observed a reduction in state space (due to the identification of packet equivalence classes) and memory footprint (due to packet/rule indexing and bit packing) from ${\sim}7$m to ${\sim}200$k states and from ${\sim}6$KB per state to ${\sim}230$B per state.
For the stateless and stateful firewall applications, resp., \toolname\ performs equally well to Kuai with respect to scaling up. 



\subsection{Model Expressivity}
\label{subsec:expressivity}

The proposed model is significantly more expressive compared to Kuai as it allows for more asynchronous concurrency. 
To begin with, in \toolname, controller messages sent before a barrier request message can be interleaved with all other enabled actions, other than the control messages sent after the barrier. By contrast, Kuai always flushes all control messages until the last barrier in one go, masking a large number of interleavings and, potentially, buggy behaviour. 
Next, in \toolname\, \emph{nomatch, ctrl} and \emph{fwd} can be interleaved with other actions. 
In Kuai, it is enforced a mutual exclusion concurrency control policy through the \emph{wait}-semaphore: whenever a \emph{nomatch} occurs the mutex is locked and it is unlocked by the \emph{fwd} action of the thread \emph{nomatch-ctrl-fwd} which refers to the same packet; all other threads are forced to wait. 
Moreover, \toolname\ does not impose any limit on the size of the \emph{rq} queue, in contrast to Kuai where only one packet can exist in it. 
In addition, Kuai does not support notifications from the data plane to the controller for completed operations as it does not support reply messages and as a result any bug related to the fact that the controller is not synced to data-plane state changes is hidden.\footnote{There are further small extensions; for instance, in MOCS the controller can send multiple \emph{PacketOut} messages (as OpenFlow prescribes).}
Also, our specification language for states is more expressive than Kuai's, as we can use any property in LTL without ``next", whereas Kuai only uses invariants with a single outermost $\Box$.\\
The MOCS extensions, however, are conservative with respect to Kuai, that is we have the following theorem (without proof, which is straightforward):

\begin{thm}[\toolname\ Conservativity]
	\label{thm:inclusiveness}
	Let 
	$\mathcal{M}_{(\lambda,\textsc{cp})}  = (S, A, \hookrightarrow, s_0, AP, L)$
	and
	$\mathcal{M}^{\mbox{\tiny{\textit{K}}}}_{(\lambda,\textsc{cp})}   = (S_K, A_K, \hookrightarrow_K, s_0, AP, L)$
	the original SDN models of \toolname\ and Kuai, respectively, using the same topology and controller.
	Furthermore, let $\mathit{Traces}( \mathcal{M}_{(\lambda,\textsc{cp})})$ and
	$\mathit{Traces}(\mathcal{M}^{\mbox{\tiny{\textit{K}}}}_{(\lambda,\textsc{cp})} )$
	denote the set of all initial traces in these models, respectively. Then,
	$
	\mathit{Traces}(\mathcal{M}^{\mbox{\tiny{\textit{K}}}}_{(\lambda,\textsc{cp})} )  \subseteq   \mathit{Traces}( \mathcal{M}_{(\lambda,\textsc{cp})})
	$.
	\end{thm}

\noindent For each of the extensions mentioned above, we briefly describe an example (controller program and safety property) that expresses a bug that is impossible to occur in Kuai.

\noindent\textbf{Control message reordering bug. }Let us consider a stateless firewall in Figure~\ref{fig:prop1} (controller is not shown), which is supposed to block incoming {\sc ssh} packets from reaching the server (see \S \ref{appendix:CP}-\ref{rq1-alg11}). Formally, the safety property to be checked here is $\Box (\forall \mathit{pkt}\, {\in}\,\mathit{S.rcvq} \,.\, \neg \mathit{pkt}.\mathrm{\textsc {ssh}})$. Initially, flow tables are empty. Switch $A$ sends a \emph{PacketIn} message to the controller when it receives the first packet from the client (as a result of a \emph{nomatch} transition). The controller, in response to this request (and as a result of a \emph{ctrl} transition), sends the following \emph{FlowMod} messages to switch $A$; rule \texttt{r1} has the highest priority and drops all {\sc ssh} packets, rule \texttt{r2} sends all packets from port 1 to port 2, and rule \texttt{r3} sends all packets from port 2 to port 1. If the packet that triggered the transition above is an {\sc ssh} one, the controller drops it, otherwise, it instructs (through a \emph{PacketOut} message) $A$ to forward the packet to $S$. A bug-free controller should ensure that \texttt{r1} is installed before any other rule, therefore it must send a barrier request after the \emph{FlowMod} message that contains \texttt{r1}. If, by mistake, the \emph{FlowMod} message for \texttt{r2} is sent before the barrier request, $A$ may install \texttt{r2} before \texttt{r1}, which will result in violating the given property. 
\toolname\ is able to capture this buggy behaviour as its semantics allows control messages prior to the barrier to be processed in a interleaved manner.

%
\ignore{
\begin{figure}[ht] 
	\centering
	\captionsetup{width=1\textwidth}
		\includegraphics[width=.45\linewidth]{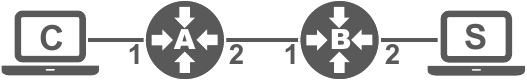} 
		\caption{A network with two switches}	
	\label{fig:prop2}
	\vspace{-.3cm}
\end{figure}
}

\vspace{-.3cm}
\begin{figure}[ht] 
	\centering
	\captionsetup{width=1\textwidth}
	\begin{subfigure}[b]{.5\textwidth}
		\captionsetup{width=.8\textwidth}
		\centering
		\includegraphics[width=1\linewidth]{figures/c_s_2.pdf} 
		\caption{}
		\label{fig:prop1}
	\end{subfigure}%
	\quad \quad
	\begin{subfigure}[b]{.4\textwidth}
		\captionsetup{width=.8\textwidth}
		\centering
		\includegraphics[width=.86\linewidth]{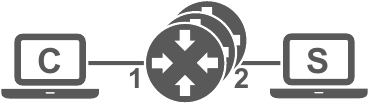}  
		\caption{}
		\label{fig:prop2}
	\end{subfigure}%
	\ignore{
		\begin{subfigure}[b]{.33\textwidth}
			\captionsetup{width=.8\textwidth}
			\centering
			\includegraphics[width=.6\linewidth]{figures/fw.pdf}  
			\caption{}
			\label{fig:prop3}
		\end{subfigure}%
	}
	\caption{Two networks with  (a) two switches, and  (b) $n$ stateful firewall replicas}	
	\label{fig:example:top}
	\vspace{-.3cm}
\end{figure}

\ignore{
\resizebox{1\textwidth}{!}{%
	\begin{minipage}{1.8\textwidth}
\begin{algorithm}[H]
	\DontPrintSemicolon
	\SetKwInOut{Input}{Input}
	\SetKwInOut{Output}{Output}
	\SetKwFunction{FMain}{PacketIn\_handler}
	\SetKwProg{Fn}{void}{:}{}
	\SetAlgorithmName{Controller Program}{}

	\Fn{\FMain{\upshape{\texttt{p}}}}{

		\uIf{\upshape{\texttt{p.prot == SSH}}}{
			\texttt{drop(p)}
		}		
		\Else{
			$\texttt{\textcolor{ForestGreen}{rule\_t} r1 = \big \{\{prio = 10\}, \{prot = SSH\}, \{in\_port = *\}, \{forward\_port = drop\}\big \}  }$ 
			
			$\mathrlap{\texttt{\textcolor{ForestGreen}{rule\_t} r2 = \big \{\{prio =  ~1\}, \{prot = *~~\},}}\hphantom{\texttt{rule\_t r1 = \big \{\{prio = 10\}, \{prot = SSH\}, }}  \texttt{\{in\_port = 1\}, \{forward\_port = 2  ~ \}\big \}  }$
			
			$\mathrlap{\texttt{\textcolor{ForestGreen}{rule\_t} r3 = \big \{\{prio =  ~1\}, \{prot = *~~\},}}\hphantom{\texttt{rule\_t r1 = \big \{\{prio = 10\}, \{prot = SSH\}, }}  \texttt{\{in\_port = 2\}, \{forward\_port = 1   ~~\}\big \}  }$\;
			
			\ForAll{\upshape{\texttt{sw}} $\in$ \upshape{\texttt{Switches}}}{
				\texttt{send\_message\_list}(\texttt{[add(r2),add(r1),b,add(r3)],sw}) 
			}	
		}		
	}
	\caption{With control messages reordering bug}
	\label{rq1-alg1}
\end{algorithm}
  \end{minipage}%
}
}

\noindent\textbf{Wrong nesting level bug. }Consider a correct controller program that enforces that server $S$ (Figure \ref{fig:prop1}) is not accessible through {\sc ssh}. Formally, the safety property to be checked here is $\Box (\forall \mathit{pkt}\, {\in}\, \mathit{S.rcvq} \,.\, \neg \mathit{pkt}.\mathrm{\textsc{ssh}})$. For each incoming \emph{PacketIn} message from switch $A$, it checks if the enclosed packet is an {\sc ssh} one and destined to $S$. If not, it sends a \emph{PacketOut} message instructing $A$ to forward the packet to $S$. It also sends a \emph{FlowMod} message to $A$ with a rule that allows packets of the same protocol (not {\sc ssh}) to reach $S$. In the opposite case ({\sc ssh}), it checks (a Boolean flag) whether it had previously sent drop rules for {\sc ssh} packets to the switches. If not, it sets flag to true, sends a \emph{FlowMod} message with a rule that drops {\sc ssh} packets to $A$ and drops the packet. Note that this inner block does not have an \texttt{else} statement.

A fairly common error is to write a statement at the wrong nesting level (\S \ref{appendix:CP}-\ref{rq1-alg2}). Such a mistake can be built into the above program by nesting the outer \texttt{else} branch in the inner \texttt{if} block, such that it is executed any time an {\sc ssh}-packet is encountered but the {\sc ssh} drop-rule has already been installed (i.e. flag \texttt{f} is true). Now, the {\sc ssh} drop rule, once installed in switch $A$, disables immediately a potential $\mathit{nomatch(A,p)}$ with $p.\mathrm{\textsc{ssh}}=true$ that would have sent packet $p$ to the controller, but if it has not yet been installed, a second incoming {\sc ssh} packet would lead to the execution of the \texttt{else} statement of the inner branch. This would violate the property defined above, as $p$ will be forwarded to $S$\footnote{Here, we assume that the controller looks up a static forwarding table before sending \emph{PacketOut} messages to switches.}.

\toolname\ can uncover this bug because of the correct modelling of the controller request queue  and the asynchrony between the concurrent executions of control messages sent before a barrier.
Otherwise, the second packet that triggers the execution of the wrong branch would not have appeared in the buffer before the first one had been dealt with by the controller. Furthermore, if all rules in messages up to a barrier were installed synchronously, the second packet would be dealt with correctly, so no bug could occur.

\noindent\textbf{Inconsistent update bug. }OpenFlow's barrier and barrier reply mechanisms allow for updating multiple network switches in a way that enables \emph{consistent packet processing}, i.e., a packet cannot see a partially updated network where only a subset of switches have changed their forwarding policy in response to this packet (or any other event), while others have not done so. \toolname\ is expressive enough to capture this behaviour and related bugs. In the topology shown in Figure \ref{fig:prop1}, let us assume that, by default, switch $B$ drops all packets destined to $S$. Any attempt to reach $S$ through $A$ are examined separately by the controller and, when granted access, a relevant rule is installed at both switches (e.g. allowing all packets from $C$ destined to $S$ for given source and destination ports). Updates must be consistent, therefore the packet cannot be forwarded by $A$ and dropped by $B$. Both switches must have the new rules in place, before the packet is forwarded. To do so, the controller, (\S \ref{appendix:CP}-\ref{rq1-alg3}), upon receiving a \emph{PacketIn} message from the client's switch, sends the relevant rule to switch $B$ (\emph{FlowMod}) along with respective barrier (\emph{BarrierReq}) and temporarily stores the packet that triggered this update. Only after receiving \emph{BarrierRes} message from $B$, the controller will forward the previously stored packet back to $A$ along with the relevant rule. This update is consistent and the packet is guaranteed to reach $S$. A (rather common) bug would be one where the controller installs the rules to both switches and at the same time forwards the packet to $A$. In this case, the packet may end up being dropped by $B$, if it arrives and gets processed before the relevant rule is installed, and therefore the invariant
$\Box \big( [\mathit{drop(pkt,sw)}]  \,.\,  \neg (\mathit{pkt.dest = S)}  \big)$, where $[\mathit{drop(pkt,sw)}]$ is a quantifier that binds dropped packets (see definition in \S \ref{appendix:CP}-\ref{rq1-alg3}), would be violated. For this example, it is crucial that \toolname\ supports barrier response messages.
\section{Conclusion}
\label{conclusion}

We have shown that an OpenFlow compliant SDN model, with the right optimisations, can be model checked to discover subtle real-world bugs. We proved that \toolname\ can capture real-world bugs in a more complicated semantics without sacrificing performance.

But this is not the end of the line. One could automatically compute equivalence classes of packets that cover all behaviours (where we still computed manually). To what extent the size of the topology can be restricted to find bugs in a given controller is another interesting research question, as is the analysis of the number and length of interleavings necessary to detect certain bugs. In our examples, all bugs were found in less than a second. 
	%
	%
	%
	\newpage
	\bibliographystyle{splncs04}
	\bibliography{ref}

\appendix
\longversion{ 
\section{Safeness}\label{appendix:proofs}

\begin{customlemma}{\ref{lemma:safe}}[Safeness]
	For an SDN network model $\mathcal{M}_{(\lambda,\textsc{cp})}  = (S, A, \hookrightarrow, s_0, AP, L)$ and context {\sc ctx}~$=(${\sc cp},~$\lambda,~\varphi)$ with $\varphi\in \text{LTL}_{ \setminus \{\bigcirc \}}$, 
	\[\alpha \in A  \text{ is safe   }\iff  \bigwedge\nolimits^3_{i=1} \mathit{Safe}_i(\alpha)  \]
	where $\mathit{Safe}_i$, given in Table \ref{tab:safeness}, are per-row. 
\end{customlemma}

\begin{proof}
To show safety we need to show two properties: \emph{independence} (action is independent of any other action) and \emph{invisibility} w.r.t.\ the context, in particular controller program, topology function and formula $\varphi$.

\paragraph{Independence:} 
Recall that two actions $\alpha$ and $\beta \neq \alpha$ are independent iff for any state $s$ such that $\alpha\in A(s)$ and $\beta\in A(s)$:
\begin{enumerate} 
\item $\alpha \in A\big(\beta(s)\big)$ and $\beta\in A\big(\alpha(s)\big)$  
\item $\alpha\big(\beta(s)\big)=\beta\big(\alpha(s)\big)$
\end{enumerate}
(1): It can be easily checked that no safe action disables any other action, nor is any safe action disabled by any other action, so the first condition of independence holds.
\\
(2): For any safe action $\alpha$ and any other action $\beta$ we can assume already that they meet Condition (1). Let us perform a case analysis on $\alpha$:

 \begin{itemize}[label=$\blacktriangleright$]
 	\item $\alpha$ is either \emph{brepl}, \emph{recv} or \emph{fwd}:\\
	To show that any interleaving with any action $\beta \neq \alpha$ leads to the same state, we observe that the changes of packet queues by these actions 
	do not interfere with each other. 
 	In cases where a packet is removed from a queue by $\alpha $ (e.g.\ $\alpha = \mathit{recv(h,pkt)}$ removes from $\mathit{h.rcvq}$) but then inserted into the same queue by $\beta$ (e.g.\ $\beta = \mathit{fwd(sw,pkt,ports)}$ where $\mathit{h \in \lambda(sw,ports)}_1$), there is no conflict either, as both actions must have been enabled in the original state in the first place. So no conflicts arise for those $\alpha$.
 	\item $\alpha$ is $\mathit{ctrl(pkt,cs)}$:
 	\begin{itemize}
 		\item If $\beta$ is not a \emph{ctrl} or \emph{bsync} action, then the same argument as above holds. 
 		\item The interesting cases occur when $\beta$ is in $\{\mathit{ctrl(\cdot), bsync(\cdot)}\}$.
 		\ignore{
 		Either the controller is insensitive to the order of any two \emph{ctrl} actions $\alpha$ and $\beta$ and we are done, or it is sensitive to some packets, but then from \emph{Safe}$_2(\alpha)$ in Table~\ref{tab:safeness} 
		 we know that $\alpha$ with given $pkt$ and $\beta$ with given $\mathit{pkt'}$, respectively, have no effect on the given controller state. 
		}
From \emph{Safe}$_2(\alpha)$ we know that {\sc cp} is not order-sensitive, which implies that $\alpha$ and $\beta$ are independent.
Order-insensitivity is a relatively strong condition but it ensures correctness of the lemma and thus partial order reduction.\footnote{Generalisations by a more clever analysis of the controller program are a future research topic.} Thus any interleaving of $\alpha$ and $\beta$ leads to the same state.

 	\end{itemize}
 \item $\alpha$ is $\mathit{bsync(sw,xid,cs)}$:\\
The same line of argument applies as for $\mathit{ctrl(pkt,cs)}$, simply exchanging the roles of $\alpha$ and $\beta$.
 \end{itemize}

 \paragraph{Invisibility}: 
 We show this for all safe actions separately:
 \begin{itemize}

 \item $\alpha = \mathit{ctrl(pk,cs)}$. 
 The only variables $\alpha$ can change are the $\mathit{controller.rq}$, $\mathit{sw'.fq}$, $\mathit{sw'.cq}$ (for some switches $\mathit{sw'}$), and the control state $\mathit{cs}$. The first three can not appear in $\varphi$ due to the definition of the specification language. In case the control state changes, $\alpha$ is invisible to $\varphi$ because \emph{Safe}$_3(\alpha)$ in Table~\ref{tab:safeness}.
\item $\alpha = \mathit{bsync(sw,xid,cs)}$. This $\alpha$ 
only affects \emph{brq}, $\mathit{sw'.fq}$, $\mathit{sw'.cq}$ (for some switches $sw'$), and the control state $\mathit{cs}$. 
We know by definition of Specification Language (\S\ref{spec}) that it cannot refer to \emph{brq} or any $\mathit{sw'.fq}$, $\mathit{sw'.cq}$.
In case the control state changes, $\alpha$ is invisible to $\varphi$ because \emph{Safe}$_3(\alpha)$ in Table~\ref{tab:safeness}.

 \item $\alpha = \mathit{fwd(sw,pk,ports)}$. Assumption \emph{Safe}$_3(\alpha)$ in Table~\ref{tab:safeness} guarantees that the only variables $\alpha$ can change, i.e.\ $\mathit{D.pq}$ or $\mathit{D.rcvq}$ for any $D$ in $\lambda(\mathit{sw,p})_1 \mid p \in \mathit{ports}$ and $\mathit{sw.pq}$, 
actually remain unchanged. Thus it follows by definition that $\alpha$ is invisible to $\varphi$.

 \item $\mathit{\alpha = brepl(sw,xid)}$. Since, by definition of Specification Language (\S\ref{spec}), the atomic propositions refer neither to any $\mathit{cq}$ nor $\mathit{brq}$, it follows from the effect of $\alpha$ that only affects $\mathit{sw.cq}$ and $\mathit{brq}$ 
 that any $\mathit{brepl(\cdot)}$ is always invisible.
\item $\mathit{\alpha = recv(h,pk)}$. Assumption \emph{Safe}$_3(\alpha)$ in Table~\ref{tab:safeness} guarantees that $\varphi$ does not refer to $\mathit{h.rcvq}$, which is the only variable affected by $\alpha$, 
and therefore $\mathit{recv(h,pk)}$ is invisible to $\varphi$.

 \end{itemize}

\end{proof}

\section{Controller Programs}\label{appendix:CP}

\setcounter{algocf}{0}
\renewcommand{\thealgocf}{CP\arabic{algocf}}
\renewcommand{\theHalgocf}{Supplement.\thealgocf}
\resizebox{1\textwidth}{!}{%
	\begin{minipage}{1.8\textwidth}
		\begin{algorithm}[H]
			\DontPrintSemicolon
			\SetKwInOut{Input}{Input}
			\SetKwInOut{Output}{Output}
			\SetKwFunction{FMain}{pktIn}
			\SetKwProg{Fn}{handler}{:}{}
			\SetAlgorithmName{\large Controller Program}{}
			
			
			%
			
			%

			\Fn{\FMain{\upshape{\texttt{sw, pkt}}}}{

				\If(\textcolor{gray}{\tcp*[f]{Otherwise, pkt is dropped silently}} ){\upshape{\texttt{\textcolor{black}{\textbf{not}} pkt.SSH}}}{
					\texttt{send\_message}\big(\texttt{PacketOut(pkt, 2), sw}\big)   \;
				}		
				
				$\texttt{rule1} \gets \texttt{\big \{\{prio} \gets \texttt{10\}, \{SSH } \gets \texttt{1~~\}, \{in\_port} \gets \texttt{*\}, \{fwd\_port} \gets \texttt{drop\}\big \}  } $\;
				
				$\mathrlap{\texttt{rule2} \gets \texttt{\big \{\{prio} \gets \texttt{~1\}, \{SSH } \gets \texttt{*~~\},}}\hphantom{\texttt{rule1 = \big \{\{prio = 10\}, \{prot = SSH\}, }}  \texttt{\{in\_port} \gets \texttt{1\}, \{fwd\_port} \gets \texttt{2  ~ \}\big \}  }$ \textcolor{gray}{\tcp*[r]{asterisk (*) matches any value}}
				
				$\mathrlap{\texttt{rule3} \gets \texttt{\big \{\{prio} \gets \texttt{~1\}, \{SSH } \gets \texttt{*~~\},}}\hphantom{\texttt{rule1 = \big \{\{prio = 10\}, \{prot = SSH\}, }}  \texttt{\{in\_port} \gets \texttt{2\}, \{fwd\_port} \gets \texttt{1   ~~\}\big \}  }$\;
				
				\ForAll(\textcolor{gray}{\tcp*[f]{Switches is the set of all switches}}){\upshape{\texttt{s}} $\in$ \upshape{\texttt{Switches}}}{
					\texttt{send\_message}\big(\texttt{FlowMod\big(add(rule2)\big), s}\big)  \;
					\texttt{send\_message}\big(\texttt{FlowMod\big(add(rule1)\big), s}\big)  \;
					\texttt{send\_message}\big(\texttt{BarrierReq\big(b\_id\big), ~~~s}\big)  \textcolor{gray}{\tcp*[r]{b\_id is a barrier identifier}}
					\texttt{send\_message}\big(\texttt{FlowMod\big(add(rule3)\big), s}\big)  \;
				}
			}
			\caption{\large A stateless firewall filter with control messages reordering bug. In a bug-free program (the one we used to verify in \S\ref{experimental-evaluation}), \texttt{rule1} should be sent first and followed by a barrier. Property: \emph{``neither host should be accessed over {\sc ssh}"}. Formally, $\Box\big(\forall h\, {\in}\, \mathit{Hosts}\,\forall \mathit{pkt}\, {\in}\, \mathit{h.rcvq} \,.\, \neg \mathit{pkt}.\mathrm{\textsc{ssh}})$.}
			\label{rq1-alg11}
		\end{algorithm}
	\end{minipage}%
}
\\\\\\
\resizebox{1\textwidth}{!}{%
	\begin{minipage}{1.8\textwidth}
		\begin{algorithm}[H]
			\DontPrintSemicolon
			\SetKwInOut{Input}{Input}
			\SetKwInOut{Output}{Output}
			\SetKwFunction{FMain}{pktIn}
			\SetKwFunction{FBar}{barrierIn}
			\SetKwProg{Fn}{handler}{:}{}
			\SetAlgorithmName{\large Controller Program}{}

			\Fn{\FMain{\upshape{\texttt{sw,  pkt}}}}{
				
				\uIf(\textcolor{gray}{\tcc*[f]{allowed\_conn is a fixed  ~~~~\textcolor{white}{a}~~~~~~~~~~~~~~~~~~~~~~~~~~~~~~~~~~~~~~~~~\hspace{2.6mm}~~~~~~~~~~~~~~~~~~~~~~~~~~~~~~~~~~~~~~~~~~~* whitelist of TCP   ~~~~\textcolor{white}{a}~~~~~~~~~~~~~~~~~~~~~~~~~~~~~~~~~~~~~~\hspace{2.6mm}~~~~~~~~~~~~~~~~~~~~~~~~~~~~~~~~~~~~~~~~~~~~~~* socket connections  \textcolor{white}{a}~~~~~~~~~~~~~~~~~~~~~~~~~~~~~~~~~~~~~~~~~\hspace{2.6mm}~~~~~~~~~~~~~~~~~~~~~~~~~~~~~~~~~~~~~~~~~~~* (host, TCP\_port) $\mapsto$ ~~~~\textcolor{white}{a}~~~~~~~~~~~~~~~~~~~~~~~~~~~~~~~~~~~~~~~\hspace{2.6mm}~~~~~~~~~~~~~~~~~~~~~~~~~~~~~~~~~~~~~~~~~~~~~* (host, TCP\_port) ~~~~\textcolor{white}{a}~~~~~~~~~~~~~~~~~~~~~~~~~~~~~~~~~~~~~~~~~~~~~~~\hspace{2.6mm}~~~~~~~~~~~~~~~~~~~~~~~~~~~~~~~~~~~~}} ){\upshape{\texttt{allowed\_conn[pkt.src][pkt.src\_TCP\_port][pkt.dest][pkt.dest\_TCP\_port]   }}}{
					
					\texttt{send\_message}\big(\texttt{PacketOut(pkt, 2), sw}\big)   \;
					
					$\mathrlap{\texttt{rule1.src}}\hphantom{\texttt{rule2.dest\_TCP\_port}} \gets 	\texttt{pkt.src} $ \;
					$\mathrlap{\texttt{rule1.src\_TCP\_port}}\hphantom{\texttt{rule2.dest\_TCP\_port}} \gets 	\texttt{pkt.src\_TCP\_port} $ \;
					$\mathrlap{\texttt{rule1.dest}}\hphantom{\texttt{rule2.dest\_TCP\_port}} \gets 	\texttt{pkt.dest} $ \;					
									  $\texttt{rule1.dest\_TCP\_port}	\gets 	\texttt{pkt.dest\_TCP\_port} $      \; 													
					$\mathrlap{\texttt{rule1.fwd\_port}}\hphantom{\texttt{rule2.dest\_TCP\_port}} \gets 	\texttt{2} $ \;
					$\mathrlap{\texttt{rule1.prio}}\hphantom{\texttt{rule2.dest\_TCP\_port}} \gets 	\texttt{2} $ \;
					
					$\mathrlap{\texttt{rule2.src}}\hphantom{\texttt{rule2.dest\_TCP\_port}} \gets 	\texttt{pkt.dest} $ \;
					$\mathrlap{\texttt{rule2.src\_TCP\_port}}\hphantom{\texttt{rule2.dest\_TCP\_port}} \gets 	\texttt{pkt.dest\_TCP\_port} $ \;
					$\mathrlap{\texttt{rule2.dest}}\hphantom{\texttt{rule2.dest\_TCP\_port}} \gets 	\texttt{pkt.src} $ \;					
									  $\texttt{rule2.dest\_TCP\_port}	\gets 	\texttt{pkt.src\_TCP\_port} $      \; 							
					$\mathrlap{\texttt{rule2.fwd\_port}}\hphantom{\texttt{rule2.dest\_TCP\_port}} \gets 	\texttt{1} $ \;
					$\mathrlap{\texttt{rule2.prio}}\hphantom{\texttt{rule2.dest\_TCP\_port}} \gets 	\texttt{2} $ \;
					
					\ForAll(\textcolor{gray}{\tcp*[f]{access rules are uniform across all switches, any of which acting as firewall replica}}){\upshape{\texttt{s}} $\in$ \upshape{\texttt{Switches}}}{						
						\texttt{send\_message}\big(\texttt{FlowMod\big(add(rule1)\big), s}\big)  \;
						\texttt{send\_message}\big(\texttt{FlowMod\big(add(rule2)\big), s}\big)  \;
						\texttt{send\_message}\big(\texttt{BarrierReq\big(b\_id\big), ~~~s}\big)  \textcolor{gray}{\tcp*[r]{b\_id is uniquely associated with an allowed connection}}
					}
					
				}		
				\Else{
					
					\texttt{send\_message}\big(\texttt{PacketOut(pkt, drop), sw}\big)   \;

					$\mathrlap{\texttt{drop\_rule.src}}\hphantom{\texttt{drop\_rule.dest\_TCP\_port}} \gets 	\texttt{pkt.src} $ \;
					$\mathrlap{\texttt{drop\_rule.src\_TCP\_port}}\hphantom{\texttt{drop\_rule.dest\_TCP\_port}} \gets 	\texttt{pkt.src\_TCP\_port} $ \;
					$\mathrlap{\texttt{drop\_rule.dest}}\hphantom{\texttt{drop\_rule.dest\_TCP\_port}} \gets 	\texttt{pkt.dest} $ \;					
									   $\texttt{drop\_rule.dest\_TCP\_port}	\gets 	\texttt{pkt.dest\_TCP\_port} $      \; 													
					$\mathrlap{\texttt{drop\_rule.fwd\_port}}\hphantom{\texttt{drop\_rule.dest\_TCP\_port}} \gets 	\texttt{drop} $ \;
					$\mathrlap{\texttt{drop\_rule.prio}}\hphantom{\texttt{drop\_rule.dest\_TCP\_port}} \gets 	\texttt{1} $ \;
					
					\ForAll{\upshape{\texttt{s}} $\in$ \upshape{\texttt{Switches}}}{						
						\texttt{send\_message}\big(\texttt{FlowMod\big(add(drop\_rule)\big), s}\big)  \textcolor{gray}{\tcp*[r]{access restrictions are uniform across all replicas}} 
					}						
				}		
			}
				\;
		\Fn{\FBar{\upshape{\texttt{sw, xid}}}}{

			\texttt{controller\_view[b\_id][sw]} $\gets$ \texttt{\textcolor{violet}{\textbf{true}}} \textcolor{gray}{\tcc*[r]{controller\_view associates installed rules (through the respective b\_id) \textcolor{white}{a}~~~~~~~~~~~\hspace{2mm}~~~~~~~~~\quad\,~~~~~~~~~\enspace~~~* for respective allowed connections with switches  \textcolor{white}{a}~~~~~~~~~~~~~~~~~~~~~~~~\hspace{2.5mm}~~~~~~~~~~}} 
		}
			\caption{\large Stateful inspection firewall (Figure \ref{fig:prop2}). The property we verify is: \emph{``a packet is never dropped by a rule in a switch if the controller is aware of a matching rule being already installed in this switch"}. Formally:
				$\Box \Big([\mathit{drop_m(pkt, sw)}]   \neg \mathit{controller\_view[pkt.src][pkt.src\_TCP\_port][pkt.dest][pkt.dest\_TCP\_port][sw]}\Big)$ \newline where $[\mathit{drop_m(pkt, sw)}]P$ is short for 		
				$[\mathit{match(sw, pkt, r)}]\big((r.\mathit{fwd\_ports}=\texttt{drop}) \Rightarrow P\big)$.
		}
			
			\label{rq1-alg12}
		\end{algorithm}
	\end{minipage}%
}
\\\\\\
\newlength{\commentWidth}
\setlength{\commentWidth}{8.5cm}
\newcommand{\atcp}[1]{\textcolor{gray}{\tcp*[r]{\makebox[\commentWidth]{#1\hfill}}}}
\resizebox{1\textwidth}{!}{%
	\begin{minipage}{1.8\textwidth}
		\begin{algorithm}[H]
			\DontPrintSemicolon
			\SetKwInOut{Input}{Input}
			\SetKwInOut{Output}{Output}
			\SetKwFunction{FMain}{pktIn}
			\SetKwProg{Fn}{handler}{:}{}
			\SetAlgorithmName{\large Controller Program}{}

			\Fn{\FMain{\upshape{\texttt{sw,  pkt}}}}{
				\If(\textcolor{gray}{\tcp*[f]{MAC\_table associates sender with a switch port}}){\upshape{\texttt{\textcolor{black}{\textbf{not}} MAC\_table[sw][pkt.src]}}   }{
					\upshape{\texttt{MAC\_table[sw][pkt.src] $ \gets$ \texttt{pkt.in\_port} }} \;			
				}		
				
				\uIf{\upshape{\texttt{MAC\_table[sw][pkt.dest]}}}{
					\texttt{send\_message}\big(\texttt{PacketOut(pkt, MAC\_table[sw][pkt.dest]), sw}\big)   		
					
					$\mathrlap{\texttt{rule.src}}\hphantom{\texttt{rule.fwd\_port}} \gets 	\texttt{pkt.src} $ \;
					$\mathrlap{\texttt{rule.dest}}\hphantom{\texttt{rule.fwd\_port}} \gets 	\texttt{pkt.dest} $ \;					
					$\mathrlap{\texttt{rule.in\_port}}\hphantom{\texttt{rule.fwd\_port}} \gets  \texttt{pkt.in\_port}$ \;
					
					$\texttt{rule.fwd\_port}	\gets 	\texttt{MAC\_table[sw][pkt.dest]} $      \; 		
					$\mathrlap{\texttt{rule.prio}}\hphantom{\texttt{rule.fwd\_port}} \gets 	\texttt{1} $ \;
					\texttt{send\_message}\big(\texttt{FlowMod\big(add(rule)\big), sw}\big)   				}		
				\Else{
					\texttt{send\_message}\big(\texttt{PacketOut(pkt, flood\textbackslash\{pkt.in\_port\}), sw}\big)   	   \textcolor{gray}{\tcp*[f]{ \text{pkt} will be flooded to all ports except incoming one}}
					
				}		
			}
			\caption{\large MAC learning application\textsuperscript{$^\dag$} for verifying absence of loops. In order to keep track of the network devices the packet passes through (i.e. the packet path history), the packet type is augmented with a history bit-field \emph{reached}, where each bit represents a visited/unvisited switch. As packets are being flooded, their history bit-field is re-written. The loop freedom property asserts that \emph{``a packet should not come back to the same switch"}. Formally,
				$\Box \big(\forall \mathit{sw}\, {\in}\,\mathit{Switches} ~\forall \mathit{pkt}\, {\in}\,\mathit{sw.pq} \,. \,\neg \mathit{pkt.reached[sw]}\big)$.}			
			\label{rq1-alg13}
		\end{algorithm}
		\small\textsuperscript{$^\dag$} \url{https://github.com/noxrepo/pox/blob/412a6adb38cb646748c8cfb657549787ab6d2e88/pox/forwarding/l2_learning.py}
	\end{minipage}%
}
\\\\\\
\resizebox{1\textwidth}{!}{%
	\begin{minipage}{1.8\textwidth}
		\MyBox{starta}{enda}{red}
		\MyBox[-3ex]{startb}{endb}{gray}
\hspace{-5mm}		\begin{algorithm}[H]
			
			\DontPrintSemicolon
			\SetKwInOut{Input}{Input}
			\SetKwInOut{Output}{Output}
			\SetKwFunction{FMain}{pktIn}
			\SetKwFunction{FBar}{barrierIn}
			\SetKwProg{Fn}{handler}{:}{}
			\SetAlgorithmName{\large Controller Program}{}
			
			
				%
				
				%

			
			\Fn{\FMain{\upshape{\texttt{sw, pkt}}}}{
				

				\uIf{\upshape{\texttt{pkt.SSH and pkt.dest == S}}}{
					\tikzmark{startb}
					\uIf(\textcolor{gray}{\tcp*[f]{f is initialised as false. pkt is dropped silently} }){\upshape{\texttt{\textcolor{black}{\textbf{not}} f}}}{
						$\mathrlap{\texttt{f}}\hphantom{\texttt{drop\_rule.fwd\_port}}  \gets   \texttt{\textbf{true}}  $\;				
						$\mathrlap{\texttt{drop\_rule.prio}}\hphantom{\texttt{drop\_rule.fwd\_port}} \gets 	\texttt{1} $ \;
						$\mathrlap{\texttt{drop\_rule.SSH}}\hphantom{\texttt{drop\_rule.fwd\_port}} \gets 	\texttt{pkt.SSH} $ \;
						$\mathrlap{\texttt{drop\_rule.dest}}\hphantom{\texttt{drop\_rule.fwd\_port}} \gets  \texttt{pkt.dest}$ \;
						
						$\texttt{drop\_rule.fwd\_port}	\gets 	\texttt{drop} $      \; 							
						
						\ForAll{\upshape{\texttt{s}} $\in$ \upshape{\texttt{Switches}}}{
							\texttt{send\_message}\big(\texttt{FlowMod\big(add(drop\_rule)\big), s}\big)  \;
							\texttt{send\_message}\big(\texttt{BarrierReq\big(b\_id\big), ~~~~~~~s}\big)   \textcolor{gray}{\tcp*[f]{b\_id is a barrier identifier}}\tikzmark{endb}
						}													
					}
					\tikzmark{starta}	
					\Else{
						\texttt{send\_message}\big (\texttt{PacketOut(pkt, 2), sw}\big)   \;
						$\mathrlap{\texttt{rule.prio}}\hphantom{\texttt{rule.fwd\_port}} \gets  \texttt{2} $ \;
						$\mathrlap{\texttt{rule.SSH}}\hphantom{\texttt{rule.fwd\_port}} \gets  \texttt{pkt.SSH} $ \;
						$\mathrlap{\texttt{rule.dest}}\hphantom{\texttt{rule.fwd\_port}} \gets  \texttt{pkt.dest} $ \;
						
						$\texttt{rule.fwd\_port} \gets \texttt{2} $	\;	
						
						\ForAll{\upshape{\texttt{s}} $\in$ \upshape{\texttt{Switches}}}{
							
							\texttt{send\_message}\big(\texttt{FlowMod\big(add(rule)\big), s}\big) \hfill \tikzmark{enda}
						}
					}
				}
				\Else{
					... \;
				}		
			}
			
			\caption{\large Wrong nesting level bug: Executing the \texttt{else}-branch - shaded red - would violate the policy that \emph{``server S (Figure \ref{fig:prop1}) should not be accessed over {\sc ssh}"}, $\Box(\forall \mathit{pkt}\, {\in}\, \mathit{S.rcvq} \,.\, \neg \mathit{pkt}.\mathrm{\textsc{ssh}})$.}
			\label{rq1-alg2}
		\end{algorithm}
	\end{minipage}%
}
\\\\\\
	\resizebox{1\textwidth}{!}{%
	\begin{minipage}{1.8\textwidth}
		\begin{algorithm}[H]
			\DontPrintSemicolon
			\SetKwInOut{Input}{Input}
			\SetKwInOut{Output}{Output}
			\SetKwFunction{FMain}{pktIn}
			\SetKwFunction{FBar}{barrierIn}
			\SetKwProg{Fn}{handler}{:}{}
			\SetAlgorithmName{\large Controller Program}{}

			\Fn(\textcolor{gray}{\tcp*[f]{Assumption: a drop-all rule with priority 0 is installed in switch B (Fig.\ref{fig:prop1})}}){\FMain{\upshape{\texttt{sw, pkt}}}}{
							
				\uIf(\textcolor{gray}{\tcc*[f]{b\_id is uniquely associated with rule\_S which \textcolor{white}{a}~~~~~~~~~~~~~~~~~~~~~~~~~~~~\,\,~~~~~~\hspace{1.4mm}~~~~~~~~~~~~~~~~~~~~~~~~~~~~* overrides the drop-all entry at B, and  \textcolor{white}{a}~~~~~~~~~~~~~~~~~~~~~~~~~~~~~~~~~~~~~~~~~~~~~~~\,\,~\hspace{1.3mm}~~~~~~~~~~~~~~* allows packets to be forwarded to S through  \textcolor{white}{a}~~~~~~~~~~~~~~~~~\hspace{1.3mm}~~~~~~~~~~~~~~~~~~~~~~~~~~~~~~~~~~~~~~~\,\,~~~~~~* port 2 \textcolor{white}{a}~~~~~~~~~~~~~~~~~~~~~~~~~~~~~~~~~\hspace{.8mm}~~~~~~~~~~~~~~~~~~~~~~~~~~~~~}}){\upshape{\texttt{pkt.dest == S  and  BarrierRes(b\_id) not received}}}  {	 									
						\If(\textcolor{gray}{\tcc*[f]{packets\_held is temporarily storing packets sent by B until consistent \textcolor{white}{a}~~~~~~~~~~~~~~~~~~~~~~~~\,\,~~~~~~\hspace{1.4mm}~~~~* update is complete \textcolor{white}{a} ~~~~~~~~~\textcolor{white}{a} ~~~~~~~~~~~ \textcolor{white}{a}~\textcolor{white}{a}~~~~~~~\hspace{1.2mm}~~~ ~~~~~~\,\,~ ~~~~\hspace{.3mm}~~~~~~~\textcolor{white}{a}}}){\upshape{\texttt{\textcolor{black}{\textbf{not}} packets\_held[sw][pkt]  }}}{
							\upshape{\texttt{packets\_held[sw][pkt] $ \gets$ \texttt{\textbf{true}} }} \;			
							$\texttt{rule\_S} \gets \texttt{\big \{\{dest} \gets \texttt{S\}, \{fwd\_port} \gets \texttt{2\}, \{prio} \gets \texttt{2\}\big \}  } $\;												
							\texttt{send\_message}\big(\texttt{FlowMod\big(add(rule\_S)\big), B}\big)   	\;
							\texttt{send\_message}\big(\texttt{BarrierReq\big(b\_id\big), ~~~~B}\big)																	
					}					
				}		
				\Else{
						\texttt{send\_message}\big (\texttt{PacketOut(pkt, 2), sw}\big)   \;
				}		
			}
		\;
			\Fn{\FBar{\upshape{\texttt{sw, xid}}}}{
					\If{\upshape{\texttt{xid == b\_id}}}{			
					$\texttt{rule\_S} \gets \texttt{\big \{\{dest} \gets \texttt{S\}, \{fwd\_port} \gets \texttt{2\}, \{prio} \gets \texttt{2\}\big \}  } $\;												
						\ForAll(\textcolor{gray}{\tcp*[f]{all switches except B}}){\upshape{\texttt{s}} $\in$ \upshape{\texttt{Switches\textbackslash\{B\}}}}{
							\texttt{send\_message}\big(\texttt{FlowMod\big(add(rule\_S)\big), s}\big)
					}

				\While{\upshape{\texttt{packets\_held[swi][p] for some (p, swi) \texttt{\textbf{and}} p.dest == S}}}{
					\upshape{\texttt{packets\_held[swi][p]}} $\gets$ \texttt{\textbf{false}}    \textcolor{gray}{\tcp*[r]{swi is the switch packet p was sent from}} 
					\texttt{send\_message}\big (\texttt{PacketOut(p, 2), swi}\big)   
				}														  																		
				}		
			}
			\caption{\large Consistent updates. We verify the property that \emph{``a packet destined to server S is never dropped at any switch"}. Formally:
				$\Box \big( [\mathit{drop_{mf}(pkt,sw)}]   \neg (\mathit{pkt.dest = S})  \big)$, where $[\mathit{drop_{mf}(pkt,sw})]P$ is  short for $[\mathit{match(sw, pkt, r)}]\big((r.\mathit{fwd\_ports}=\texttt{drop}) \Rightarrow P\big) \land [\mathit{fwd}(\mathit{sw, pkt}, \mathit{fwd\_ports})]\big((\mathit{fwd\_ports}=\texttt{drop}) \Rightarrow P\big).$
							}
			\label{rq1-alg3}
		\end{algorithm}
	\end{minipage}%
}

\newpage

}
	
\end{document}